\documentclass[10pt]{IEEEtran}

\usepackage{bm}
\usepackage{cite}
\usepackage{psfrag}
\usepackage{latexsym, amsmath, color, amsfonts, amssymb,graphicx}
\usepackage{algorithm}
\usepackage{algorithmic}
\usepackage{epic}
\usepackage{pgf,tikz,xcolor}
\usepackage{rotating}
\usepackage{mathrsfs}

\newtheorem{theorem}{Theorem}%[section]
\newtheorem{lemma}{Lemma}
\newtheorem{corollary}{Corollary}

\newtheorem{remark}{Remark}

\newtheorem{example}{Example}

\newtheorem{problem}{Problem}

\newcommand{\expt}{\bm{E}}

\DeclareMathOperator*{\argmin}{argmin}

\title{ Privacy Preservation by Local Design  %for Parties
 in Cooperative Networked Control Systems} %Non-Collaborative Privacy Preservation
\author{Chao Yang$^\star$, Yuqing Ni$^\dag$, Wen Yang$^\star$, and Hongbo Shi$^\star$
\thanks{$^\star$: Key Laboratory of Smart Manufacturing in Energy Chemical Process, Ministry of Education; Dept. of Automation, East China University of Science and Technology, Shanghai, China, 200237.
Email: \{yangchao, weny, hbshi\}@ecust.edu.cn.
$\dag$: Key Laboratory of Advanced Process Control for Light Industry (Ministry of Education); School of Internet of Things Engineering, Jiangnan University, Wuxi, China.
Email: yuqingni@jiangnan.edu.cn.
The corresponding author is Wen Yang.} 
}

\begin{document}
\maketitle
                  
\begin{abstract} 
In this paper, we study the privacy preservation problem in a cooperative networked control system, which has closed-loop dynamics, working for the task of linear quadratic Guassian (LQG) control. The system consists of a user and a server:
the user owns the plant to control, while the server provides computation capability, and the user employs the server to compute control inputs for it. 
To enable the server's computation, the user needs to provide the measurements of the plant states to the server, who then calculates estimates of the states, based on which the control inputs are computed.
However, the user regards the states as privacy, and makes an interesting request: the user wants the server to have ``incorrect" knowledge of the state estimates rather than the true values. 
Regarding that, we propose a novel design methodology for the privacy preservation, in which the privacy scheme is locally equipped at the user side not open to the server, which manages to create a deviation in the server's knowledge of the state estimates from the true values.
However, this methodology also raises significant challenges: in a closed-loop dynamic system, when the server's seized knowledge is incorrect, the system's behavior becomes complex to analyze; even the stability of the system becomes questionable, as the incorrectness will accumulate through the closed loop as time evolves. 
In this paper, we succeed in showing that the performance loss in LQG control caused by the proposed privacy scheme is bounded by rigorous mathematical proofs, which convinces the availability of the proposed design methodology.
% In this paper, we succeed in showing that the proposed privacy scheme causes bounded performance loss in LQG control by rigorous mathematical proofs. 
We also propose an associated novel privacy metric and obtain the analytical result on evaluating the privacy performance.
%Furthermore, we analyze the service loss in the LQG control performance caused by the privacy scheme.
Finally, we study the performance trade-off between privacy and control, where the accordingly proposed optimization problems are solved by numerical methods efficiently.

\end{abstract}


%%%%%%%%%%%%%%%%%%%%%%%%%%%%%%%%%%%%%%%%%%%%%
\section{Introduction}\label{section:introduction}

The networked control systems, featuring the flexibility of configuring the various components of a closed-loop control system, have gained numerous applications nowadays \cite{TheSurvey}. 
This flexibility could further facilitate the system to be run cooperatively by more than one party: different companies or organizations own respective parts of it.
This could be realized by that various parties unite together, or that the original owner of a system outsources inefficient parts to external parties.
We name this type of systems as \emph{cooperative networked control systems}.
This cooperative manner should be a natural trend of economic demand, since higher economic efficiency is usually achieved by deeper work division and cooperation \cite{Freidman2005}. 
Practically, we observe the exploding cloud service in industry in recent years \cite{AWS, Azure, GoogleCloud, Aliyun}, which may be viewed as a sign of this trend, since the cloud service can be expected to provide one way easily realizing the configurations of cooperative networked control systems. 

%Practically, this trend has been proved by the exploding cloud service in industry in recent years, which encourages the realizations of various cooperative networked control systems, such as the cloud-based control systems \cite{}.
%This trend has indeed appeared in applications of networked control systems, such as the cloud-based control system, where the users rely on an external cloud service provided by cloud companies \cite{}.

A basic scenario of cooperative networked control systems is presented in Fig. \ref{fig:user-server}, which we name as the \emph{user-server system}.
In this scenario, the cooperative networked control system is formed by two parties, the user and  the server.
The user employs the server to provide certain service for it, and to achieve this, it provides necessary information to the server.

\begin{figure}[htbp]
\vspace{-3mm}
\centering
\includegraphics[scale=0.8]{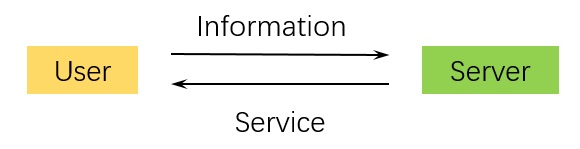}
\vspace{-3mm}
\caption{The user-server system.}
\label{fig:user-server}
\end{figure}
 
Although the cooperation of the user and the server brings efficiency and economics advantages,  it also induces the privacy issue. 
The data provided by the user possibly correspond to certain other information, some of which may be supposed to be private. 
Meanwhile, the server is motivated to analyze the information data purposely, since more knowledge about clients usually benefits more in business. 
%Moreover, in recent years, the sense and ability of collecting and mining data have exploded.
Hence, the user may face unwilled exposure of privacy. %(even unexpected sensitive information).
Take a daily-life example as an analogy. When people shop on-line, the data of items bought and money paid (maybe living address also included) are recorded.
%For example, when people shop on-line, what they leaves to the e-shop are only their shopping records about the items bought and money paid (or even living address included). 
These data can be related to the customer's income level, family structure, hobbies, and so on. 
The business provider could use the customer's historical shopping records to depict his/her preference profile, so that better business decisions can be made (such as a personalized recommendation service). 
%Meanwhile, it also means that those corresponding personal information of customers are analyzed and recovered.
However, for certain customers, they may not want their preferences exposed. 
Nevertheless, the risk seems inevitable. %### is the logic of succeeding statements ok?
We have to say that the server does not literally betray the cooperation; it is referred to as ``honest but curious" \cite{Goldreich2009}, i.e., it still honestly serves the duty in cooperation, only meanwhile intending to discover more knowledge of the user as well.

Consequently, to safely use the cooperative networked control systems, we need to address the privacy issue. 
One common methodology is that the server provides certain privacy policy in favor of the user. 
%Practically, users may have different demands in privacy preservation, and thus we should comprehend the user's demand and design proper privacy scheme for it.
%%may have different tolerances on privacy exposure and hence different demands in privacy preservation level, 
However, the preservation for privacy in this situation largely depends on the server's integrity.
Hence, a careful user may still have this demand: regardless of the server's privacy policies, it demands a \emph{local} privacy scheme whose existence is unknown to the server, so that the server's achievable knowledge of the user's private information is ``incorrect". 
%Meanwhile, it accepts to afford certain caused loss in service quality.
%When considering privacy, we have to take the privacy protection methods claimed by the data-receiving parties (such as encryption) as untrusted. 
%The data-providing parties should rely on themselves to achieve privacy preservation. 
%Hence, they demand corresponding technique support. 
According to the current literature, this methodology for privacy preservation in the direction of cooperative networked control systems has never been visited yet.
In this paper, we study this methodology.

%-----------------------------------------------------------------------------------------------------------------------
\subsection{Related Studies}

%This cooperative privacy problem has been visited in many directions, such as cloud security and outsourcing security \cite{Wang2011INFOCOM, Xiao2012CST, Ren2012IC}.

When specifically considering the privacy issue in networked control systems, one may first relate it to the study on the security of cyber-physical systems (CPSs) \cite{Mo2012IEEE, Krishnamurthy2018TIFS, Agarwal2018CDC}, since both of them corresponds to system safety. 
However, the study on privacy issue cannot be entirely covered by the methodologies in CPS security and has its own particularities due to their different natures.
1) In the CPS security problem, an adversary is excluded in the original healthy system, usually unknown. In the privacy problem, the one to defense is the known cooperator who also belongs to the system. 
2) In the CPS security problem, the information collection (usually by stealthy) is illegal and may be followed by system damaging, which is unwanted. In the privacy problem, information is collected with right authorized, and the cooperator normally serves the duty in cooperation, which is needed.
%Hence, though overlapped in some aspects, the methodology of cyber-physical security cannot entirely cover the study on privacy preservation, and the study of privacy preservation needs to develop its own one.
%Hence, though overlapped in some aspects, the methodology on privacy preservation differs from the one of the cyber-physical security and have its own particularities. 

The existing studies related to privacy in cooperative networked control systems have several frameworks, each of which is based on the understanding and modeling of the concept of privacy.
The first important one is called \emph{differential privacy}. 
This concept was originally proposed in the field of statistical databases \cite{Dwork2006ICALP, Dwork2008survey}, where it pointed out that one database item has the risk of being identified by analyzing the differences between certain answering outputs of database queries, which is later applied to vairous practical applications \cite{Ye2022TIFS}, especially in the studies on auction mechanism \cite{McSherry2007FOCS, Li2020TIFS}.
To protect the privacy, it proposed the method of adding randomness to query answers.
%the removal or addition of a single database item does not change the answering output of a same query too much, and then less risk of privacy exposure exists for participants to join the database.
%such that each single database item cannot be identified by analyzing the query answers, and hence no risk of privacy exposure exists for participants when joining the database.
This concept and methodology were later borrowed by the control field  \cite{LeNy2014TAC}. 
They were mainly employed in the model who has a \emph{centralizing} or \emph{fusing} mechanism with \emph{multiple} participating parties, which are similar to the model of a statistical database who assembles data from multiple participants. 
%
%###### MAY REVISE LATER
In this field, differential privacy was newly considered in various previous problems and models, including consensus in centralized and distributed system \cite{Huang2012WPES}, distributed estimation \cite{He2018TIT}, centralized Kalman filtering \cite{LeNy2014TAC}, linear dynamic system  \cite{Koufogiannis2017CDC}, MIMO system \cite{LeNy2018TAC}, point distribution estimation \cite{Ye2018TIT}, constrained optimization for distributed users \cite{Han2017TAC, Hale2018TCNS}, cloud-based control \cite{Hale2018ACC}.
%######
%Huang et al. \cite{Huang2012WPES} considered the model of consensus in a centralized system and in a distributed system, both with preserved differential privacy.
%He et al. \cite{He2018TIT} studied a distributed estimation model under differential privacy requirement.
%Le Ny and Pappas \cite{LeNy2014TAC} studied the differential privacy problem for centralized Kalman filtering. 
%Le Ny and Mohammady \cite{LeNy2018TAC} considered the MIMO system.
%Koufogiannis and Pappas \cite{Koufogiannis2017CDC} considered the differential privacy preservation for current state of a running dynamical system.
%Ye and Barg \cite{Ye2018TIT} considered the point distribution estimation from privatization samples.
%The model of constrained optimization for distributed multiple users was also visited, such as Han et al. \cite{Han2017TAC} and Hale and Egerstedt \cite{Hale2018TCNS}.
%Han et al. \cite{Han2017TAC} investigated the distributed optimization problem of resource allocation for participating users.

The second important framework of privacy preservation should be the \emph{information-theoretic} approach \cite{Nekouei2019Survey}.
In this framework, information entropy is used as a metric of privacy. As the indicator of how much information is carried, information entropy is a natural metric to measure the privacy.
Compared with differential privacy, who is only available for centralized models with multiple participants, information-theoretic approach also works for models with one single participant. 
% static vs dynamic!
%Both the two frameworks realize the privacy preservation by noise injection; 
Moreover, in differential privacy, the preservation is usually realized by injecting Gaussian or Laplacian noise, while in the information-theoretic approach, more sophisticated schemes are  also proposed.
%the noise could have more general distribution (but mainly works for systems with discrete state variables).
%###### TO REVISE LATER
Nekouei et al. \cite{Nekouei2018ACC} presented two privatized information-sharing schemes in a multi-sensor estimation system.
Li et al. \cite{Li2018TIT} proposed the privacy-aware policy for resource distribution by solving an optimization problem.
Nekouei et al. \cite{Nekouei2018a} designed a novel estimator which guarantees the local parameter's privacy.
Jia et al. \cite{Jia2017ICCPS} studied a closed-loop predictive control system.
Tanaka et al. \cite{Tanaka2017ACC} considered a cloud-based LQG control system.
%######

Another framework worth attention is the \emph{homomorphic encryption}. In the previous two frameworks, privacy is preserved by adding noise. Homomorphic encryption provides an entirely different methodology.
%###### TO REVISE LATER
It finds proper encryption operators or algorithms, which guarantees the homomorphic nature between the data before and after the encryption operation.
Farokhi et al. \cite{Farokhi2016NECSYS, Farokhi2020} and Tran et al. \cite{Tran2020CEP} considered such problems.
%######
Ni et al. \cite{Ni2021TIFS} proposed a novel privacy scheme for distributed estimation based on the technique of homomorphic encryption.

Moreover,  several significant studies are not included in the previous frameworks, such as \cite{Mo2014CDC, Mo2017TAC}, in which novel schemes of privacy preservation for initial values in consensus systems were proposed.
%A class consisting of a large number of studies is the consensus privacy problem. 
%Mo et al. \cite{Mo2014CDC, Mo2017TAC} proposed a scheme to protect the privacy of initial value of consensus participants.

%strategic communication; parameter

%------------------------------------------------------------------------------------------------------------------------
\subsection{Observation on the Related Studies}

After reviewing the literature, one can find a common issue in the methodologies for privacy preservation presented in existing studies: in the cooperation, the server (or the party equivalent to the server) usually has either full or partial knowledge about the privacy scheme used.
In most studies on differential privacy, the privacy scheme is designed by default at the server side, rather than at the user side 
%{\color{red}\{could talk about the centralized config., output perturbation\}
\cite{Han2017TAC, LeNy2018TAC}. 
For the cases the user locally equips a privacy scheme, the server also knows the existence of the user's privacy scheme and is usually provided the associated parameters as well \cite{Tanaka2017ACC, Hale2018ACC, Yazdani2023TAC}. In both cases, the server is able to seize the correct, or not accurate but correct, information about the user's privacy.
Then the user's privacy's safety largely depends on the server's integrity, which shows the weakness of existing methodologies.
%###### 
%differential privacy, input perturbation is for the user; but center knows, co-design
%information-theoretic: not all centralized, but co-design
%
Moreover, the privacy schemes designed by homomorphic encryption, which are nevertheless entirely on behalf of the user, have strong limitations in required assumptions and cannot be easily extended to general systems.
%
% 4) no feedback, dynamic
% A number of studies considered open-loop systems

Secondly, most existing studies investigated open-loop systems without feedback, as the analysis becomes difficult when considering a privacy scheme added in a closed-loop dynamic system. 
Only a few studies are  on the closed-loop dynamic system \cite{Jia2017ICCPS, Hale2018ACC, Yazdani2023TAC, Tanaka2017ACC, Alexandru2019ICCPS}.

%either rely on strong assumptions or system scale limitation and have no general results, or propose privacy schemes also open to the server.

~~
%%The literature \cite{Jia2017ICCPS, Hale2018ACC, Tanaka2017ACC} worked on closed-loop control systems, in which \cite{Hale2018ACC, Tanaka2017ACC} considered the LQG control systems.

%Jia et al. \cite{Jia2017ICCPS} considered a predictive control system, which is quite different from the proposed model.
%Hale et al. \cite{Hale2018ACC} considered the differential privacy problem in a cloud-based LQG control system, but the obtained results relies on strong assumptions.
%Tanaka et al. \cite{Tanaka2017ACC} also considered a cloud-based LQG control system but under the information-theoretic approach, and gave significant results on this problem. However, in their model, since the server knows the parameters of the privacy scheme, it is able to obtain a correct estimation of the state.
%Alexandru and George J. Pappas \cite{Alexandru2019ICCPS} used the method of homomorphic encryption, where the considered system scale is limited due to the computation complexity, though it can efficiently protect the privacy.

%------------------------------------------------------------------------------------------------------------------------
\subsection{The Study of This Paper}

In this paper, we specifically consider the privacy preservation problem in a user-server system who is a cooperative closed-loop dynamic control system, with the goal of the optimal LQG control.
In the system, the plant is owned by the user, while the computation function belongs to the server, and the user employs the server to compute the optimal LQG control inputs for it.
To enable the server's computation, the user has to provide measurements of the plant states to the server. 
Given the measurements, the server is able to know the state estimates and then computes the control inputs based on them.

However, in this scenario, the user takes its state trajactory as private information.
In this paper, the user raises the following request: it demands a localized privacy scheme which is unknown to the server, so that the server could only obtain ``incorrect" estimates of the states. 
We consider the design and analysis of this type of privacy preservation.

%Compared with previous methodologies, this paper avoids the situation that the server has knowledge of the employment of the privacy schemes.
%Meanwhile, this local design of privacy preservation is based on the environment of the closed-loop systems, which raises additional challenges. 
%As the closed-loop dynamic system is considered, the server's incorrect knowledge will make the system behavior complicated to analyze. 
%Moreover, as the incorrectness will accumulate through the closed loop as time evolves, it may cause instability of the system dynamics.

The novelty and main contributions of this paper are summarized as follows.
\begin{enumerate}
\item
In this paper, we propose a novel methodology for privacy preservation, where the privacy scheme is employed locally at the user side, which makes the server's knowledge of the privacy information ``incorrect".
Compared with previous methodologies, this paper avoids the situation that the server has knowledge of the employment of the privacy schemes.

%Meanwhile, this methodology is considered in the background of the closed-loop dynamic control system, which increases the difficulty in analysis.

\item
Meanwhile, this local design of privacy preservation is based on the environment of the closed-loop systems, which raises additional challenges. 
As the closed-loop dynamics is considered, the server's incorrect knowledge will make the system's behavior complex to analyze. 
Moreover, as the incorrectness will accumulate through the closed loop as time evolves, it may cause \emph{instability} of the control dynamics.
In this paper, we successfully show by rigorous mathematical proofs that the performance loss in LQG control caused by the proposed privacy scheme is bounded, which convinces the availability of the proposed design methodology.

%We analyze the privacy quality (Theorem \ref{thm:LQG2-PrivacyPerformance}) and the LQG control performance (Theorem \ref{thm:LQG2-obj}), and study the trade-off problems between them (Problem \ref{prblm:optimization2}).  

\item
We also propose a novel privacy metric associated with the \emph{state estimates}, which measures the deviation between the obtained estimates by the server and the true value. Compared with the metrics in differential privacy (the $\epsilon$-differential privacy) and informtion-theoretic approaches (information-entropy-based ones), our metric is compatible with the conventional notion in control systems and hence provides intuitive understanding of the level of privacy preservation, as in control studies, the estimates are the typical quantities considered.

%We also propose an associated novel privacy metric and obtain the analytical result on evaluating the privacy performance.

\item
Moreover, we make a primary attempt to find the localized privacy schemes \emph{without any loss in control performance}, which should be the perfect scenario for the user.
%perfect privacy

\end{enumerate}

The remainder of the paper is organized as follows. Section II proposes the study framework of the local design for privacy preservation in the considered system. Section III presents the main results. Section IV considers the privacy schemes which achieves privacy preservation while has no loss in service performance. Section V presents simulation examples. At last, Section VI discusses the possible future work directions and Section VII makes the conclusion of this study.

~\\
\textit{Notations}: 
$\mathbb{Z}_{+}$ is the set of non-negative integers and $k\in \mathbb{Z}_+$ is the time index. 
%$\mathbb{N}$ is the set of positive integers. 
$\mathbb{R}$ is the set of real numbers. 
$\mathbb{R}^{n}$ is the $n$-dimensional Euclidean space. 
$\mathbb{S}_{+}^{n}$ (and $\mathbb{S}_{++}^{n}$) is the set of $n$ by $n$ positive semi-definite matrices (and positive definite matrices); when $X \in \mathbb{S}_{+}^{n}$ (and $\mathbb{S}_{++}^{n}$), it is written as $X \geq 0$ (and $X > 0$). 
$X \geq Y$ if $X - Y \in  \mathbb{S}_{+}^{n}$. 
$\bm{E}(\cdot)$ or $\bm{E}[\cdot]$ 
is the expectation of a random variable and $\bm{E}(\cdot|\cdot)$ or $\bm{E}[\cdot|\cdot]$ is the conditional expectation. 
$\mathrm{tr}(\cdot)$ is the trace of a matrix. 
%$\lfloor x \rfloor$ denotes the largest integer which is smaller or equal to $x$. 
%For functions $f, f_1, f_2$ with appropriate domains, $f_1f_2(x)$ stands for the function composition $f_1\big(f_2(x)\big)$, and $f^{n}(x) \triangleq f\big(f^{n-1}(x)\big)$ with $f^{0}(x) \triangleq x$.}
For a matrix $X$, $\mathscr{N}(X)$ is the kernel space of $X$.

%%%%%%%%%%%%%%%%%%%%%%%%%%%%%%%%%%%%%%%%%%%%%
\section{Framework of the Privacy Preservation by Local Design}

%\begin{figure}[htbp]
%\vspace{-3mm}
%\centering
%\includegraphics[scale=0.5]{Fig-framework.jpg}
%\vspace{-3mm}
%\caption{Framework.}
%\label{fig:}
%\end{figure}

%Privacy-Preserving LQG Control

In this section, we propose the framework of privacy preservation by local design in a user-server system for the task of cooperative LQG control. 
In the remainder of this section, we present the setup of the ordinary cooperative system, the methodology of the proposed localized privacy preservation, and feasible problems to study.

%-----------------------------------------------------------------------------------------------------------------------
\subsection{Ordinary Cooperative System}

The ordinary user-server system for cooperative LQG control is illustrated in Fig. \ref{fig:LQG2original}.
The user owns a plant to control, the server owns computation capability, and the user employs the server to %work as the estimator and controller to 
compute the optimal LQG control inputs for it.
% #this para seems not necessary

\begin{figure}[htbp]
\vspace{-1mm}
\centering
\includegraphics[scale=0.4]{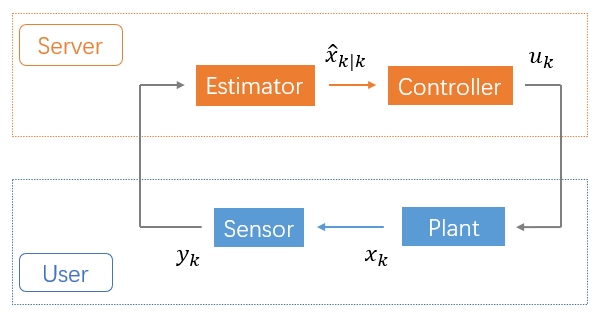}
\vspace{-3mm}
\caption{The ordinary cooperative LQG control system.}%, in which the user's state information is not accessible directly, and is measured by a sensor.}
\label{fig:LQG2original}
\end{figure}

The user has the following dynamic process to control:
\begin{eqnarray}\label{eqn:state-process}
x_{k+1}&=&Ax_k+Bu_k+w_k, ~~~%k=0,1,...,N-1,
\end{eqnarray}
where $x_k\in \mathbb{R}^n$ is the state of the process, $u_k\in \mathbb{R}^l$ is the control input, and $w_k$ is the Gaussian white noise with distribution $\mathcal{N}(0,Q) (Q\geq 0)$. 
The initial condition of the state is assumed as that $x_0$ is Gaussian with $\mathcal{N}(\bar{x}_0,\Sigma_0) (\Sigma_0\geq 0)$.
The time horizon of the process is assumed to be a finite $T$. 
$(A, B)$ is assumed to be controllable.

The user has no direct access to the value of the state $x_k$, and instead, it owns a sensor measuring the state as follows: 
\begin{eqnarray}\label{eqn:measurement-process}
y_k&=&Cx_k+v_k,
\end{eqnarray}
where $y_k\in \mathbb{R}^m$ is the measurement of $x_k$ taken by the sensor and $v_k$ is the white Gaussian noise corrupting the measurement with distribution $\mathcal{N}(0,R)$ ($R>0$). 
%The measurement noise $v_k$ is white Gaussian with distribution $\mathcal{N}(0,R)$ ($R>0$). 
The noise $\{w_k\}$, the noise $\{v_k\}$, and the initial state $x_0$ are mutually independent of each other.
We also assume that $(A, \sqrt{Q})$ is stabilizable and $(C, A)$ is detectable.

The user demands the optimal LQG control for its state process. 
The considered finite-time quadratic objective for the LQG control is denoted as $\mathcal{O}_{0:T}$ and is defined as
\begin{eqnarray} \label{eqn:obj2}
\mathcal{O}_{0:T}\triangleq \bm{E}\Big[\sum_{k=0}^{T-1}(x_k'Wx_k+u_k'Uu_k)+x_T'Wx_T\Big],
\end{eqnarray}
where $W$ and $U$ are weight matrices satisfying $W\geq 0$ and $U>0$,
%\begin{eqnarray}
%\mathcal{I}\triangleq \lim_{T\rightarrow\infty}{1\over T}\bm{E}\Big[\sum_{k=0}^{T-1}(x_k'Wx_k+u_k'Uu_k)+x_T'Wx_T\Big].
%\end{eqnarray}
and the expectation is taken with respect to all possible randomness. We also assume that $(\sqrt{W}, A)$ is detectable.

In this cooperative system, the server computes the optimal LQG control inputs for the user.
Since the process states are unavailable, the server has to also estimate the states before computing the control inputs.
That is to say, the server serves as the estimator and the controller in the control loop (Fig. \ref{fig:LQG2original}).
%The user needs to obtain the optimal control inputs according to the LQG objective function (\ref{eqn:obj2}), which are computed by the server.
%
%It does not do the computation by itself but to employ a server, who has computation capability, to provide the computing service for it. 
%

%\begin{example}
%\color{blue}
%process: a vehicle, input sent by a remote controller
%
%controller: center computer
%
%sensor: camera
%
%\end{example}

To enable the server's computation, before the process, the user shares the following system parameters with the server:
\begin{itemize}
\item
$A, B, W, U$ (for LQG control); 
\item
$C, Q, R$ (for state estimation); 
\item
time horizon $T$;
\item
initial condition $\mathcal{N}(\bar{x}_0,\Sigma_0)$. %(could also has deviation)
\end{itemize}
When the process begins, the user is also supposed to provide $y_k$ to the server at each time $k$.

This is the classic problem of LQG control where only the measurements of states are available. 
The server works as follows. % 1) to do the estimation, 2) to compute the control law based on the estimates.
Firstly, before the process, it does the following computation in favor of the optimal LQG control:
\begin{eqnarray} 
S_T\!\!\!\!&=&\!\!\!\! W, \label{eqn:LQG-S-terminal}\\
\nonumber S_k\!\!\!\!&=&\!\!\!\!A'S_{k+1}A+W\\
 &&\!\!\!-A'S_{k+1}B(B'S_{k+1}B+U)^{-1}B'S_{k+1}A, \label{eqn:LQG-S-recursion}\\ 
L_k\!\!\!\!&=&\!\!\!\!-(B'S_{k+1}B+U)^{-1}B'S_{k+1}A. \label{eqn:LQG-gain}
\end{eqnarray}

Secondly, it estimates the user's state $x_k$, where the minimum mean-squared error (MMSE) estimates are calculated. 
When the process begins, at each time $k$, the user sends $y_k$ to the server. Let $\bm{Y}_k\triangleq \{y_1,y_2,...,y_k\}$. Then the server computes the \textit{a priori} and \textit{a posteriori} estimates $\hat{x}_{k|k-1}$ and $\hat{x}_{k|k}$ defined as follows:
\begin{eqnarray*}
\hat{x}_{k|k-1} &\triangleq& \bm{E}[x_k|\bm{Y}_{k-1}],\\
\hat{x}_{k|k}   &\triangleq& \bm{E}[x_k|\bm{Y}_{k}].
\end{eqnarray*}
Meanwhile, let $P_{k|k-1}$ and $P_{k|k}$ be the estimation error covariance matrices associated with $\hat{x}_{k|k-1}$ and $\hat{x}_{k|k}$, respectively:
\begin{eqnarray*}
P_{k|k-1} &\triangleq& \bm{E}[(x_k-\hat{x}_{k|k-1})(x_k-\hat{x}_{k|k-1})'|\bm{Y}_{k-1}],\\
P_{k|k}   &\triangleq& \bm{E}[(x_k-\hat{x}_{k|k})(x_k-\hat{x}_{k|k})'|\bm{Y}_{k}].
\end{eqnarray*}
The calculation of the state estimation is standard: by Kalman filtering. 
Given the initial value
$\hat{x}_{0|0}=\bar{x}_0$ and $P_{0|0}=\Sigma_0$, for $k\geq 1$, the server does the following computation:
\begin{eqnarray}
\hat{x}_{k|k-1} & = & A\hat{x}_{k-1|k-1}+Bu_{k-1}, \label{eqn:state-time-update} \\
P_{k|k-1} & = & AP_{k-1|k-1}A' + Q, \label{eqn:error-covariance-time-update} \\
K_{k} & = & P_{k|k-1}C'(CP_{k|k-1}C' + R)^{-1}, \label{eqn:kalman-gain} \\
\hat{x}_{k|k} & = & \hat{x}_{k|k-1} + K_{k}(y_{k} - C\hat{x}_{k|k-1}), \label{eqn:state-measurement-update} \\
P_{k|k} & = &(I -K_{k}C)P_{k|k-1}. \label{eqn:error-covariance-measurement-update}
\end{eqnarray}

Thirdly, the server computes the optimal control input:
\begin{eqnarray*}
u^\star_k=L_k \hat{x}_{k|k},
\end{eqnarray*}
and returns $u^\star_k$ to the user's actuator.

Then we evaluate the performance of the optimal LQG control.
The optimal objective value, denoted as $\mathcal{O}_{0:T}^\star$, is given as follows \cite{Betsekas1995}:
\begin{eqnarray} \label{eqn:LQG-obj-opt}
\mathcal{O}_{0:T}^\star=\bm{E}(x_0'S_0x_0)+\sum_{k=0}^{T-1}\mathrm{tr}(S_{k+1}Q)+\sum_{k=0}^{T-1}\mathrm{tr}(\mathit\Phi_{k}P_k),
\end{eqnarray}
where 
\begin{eqnarray*}
\mathit\Phi_k=A'S_{k+1}B(B'S_{k+1}B+U)^{-1}B'S_{k+1}A.
\end{eqnarray*}

%-----------------------------------------------------------------------------------------------------------------------
\subsection{Local Design of Privacy Preservation} 

\begin{figure}[htbp]
\vspace{-3mm}
\centering
\includegraphics[scale=0.35]{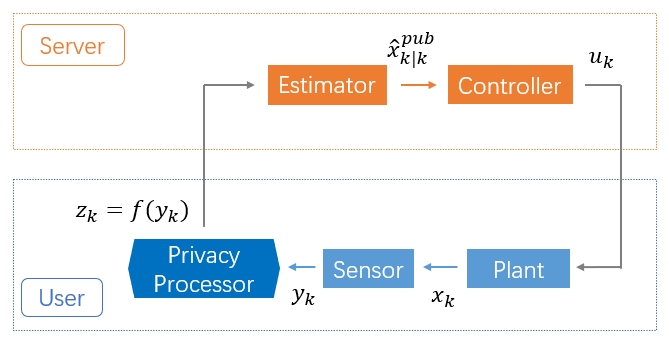}
\vspace{-3mm}
\caption{A local privacy processor is used.}
\label{fig:LQG2private}
\end{figure}

In the ordinary cooperative system, the estimates of the user's states trajectory $\{x_k\}$ is transparent to the server. 
However, the user takes its states trajectory as privacy.
It wants to conceal the true values of state estimates and let the server have \emph{incorrect} state estimates, which are deviated from the true ones.
Hence, it decides to employ a local privacy scheme, not open to the server, to achieve the desired privacy preservation.

In this paper, as a start, we consider the simplest design of the privacy scheme for the user: just adding a module of privacy processor to process the data to be sent to the server (Fig. \ref{fig:LQG2private}).
At each time $k$,  the user generates a signal $z_k$, and we let it be a function of $y_k$:
\begin{eqnarray*}
z_k=f(y_k),
\end{eqnarray*}
and sends it instead of $y_k$ to the server.
The function $f(\cdot)$ could generally depend on $\bm{Y}_k$ and $u_0,u_1,...,u_{k-1}$:
\begin{eqnarray*}
z_k=f(\bm{Y}_k,u_0,u_1,...,u_{k-1}).
\end{eqnarray*}

Specifically, in this paper, we consider one basic scheme: noise injection. We add to $y_k$ a noise signal $\delta_k$, say,
\begin{eqnarray}\label{eqn:LQG2-scheme}
z_k=y_k+\delta_k.
\end{eqnarray}
The noise $\delta_k$ is assumed to be i.i.d. and to have the Gaussian distribution $\mathcal{N}(0, \Sigma_\delta) ( \Sigma_\delta\geq 0)$.

%----------remark------------
\begin{remark} \itshape
In existing related studies \cite{Tanaka2017ACC, Hale2018ACC, Yazdani2023TAC}, the privacy scheme is also located at the user side, but the existence of the privacy scheme is known by the server and the parameters are also shared with the server.
Although the server cannot have access to the accurate state values because of the privacy scheme, it is still able to obtain correct estimation based on the shared knowledge of the privacy scheme.

%Compare with Tanaka2017ACC. Compute the states estimate according to the distribution knowledge.

\end{remark}
%%----------end remark------------

%-----------------------------------------------------------------------------------------------------------------------
\subsection{Privacy Metric} 

Since the server is unaware that the user is using a local privacy scheme, it will take $z_k$ as the user's measurement $y_k$.
Then the signal $z_k$ will create deviations in the server's state estimates, as the user demands.

To analyze the quality of privacy preservation in this scenario, we clarify the following notations.
\begin{itemize}
\item
The server's estimate of $x_k$ and the associated error covariance are denoted as $\hat{x}^{svr}_{k|k}$ and $P^{svr}_{k|k}$.
This estimate is not true but deviates from the true $\hat{x}_{k|k}$.

\item
The estimate of $x_k$ based on the knowledge of $y_k$, i.e., the true estimate, and the associated error covariance, are still denoted by $\hat{x}_{k|k}$ and $P_{k|k}$.

%By sending the processed signals, the user makes the server impossible to have the correct information of the estimates trajectory.

\end{itemize}
Since the server is planned to compute the state estimates and the associated error covariance by eqn. (\ref{eqn:state-time-update})-(\ref{eqn:error-covariance-measurement-update}),  the server's estimate and error covariance evolve accordingly:
\begin{eqnarray}
\hat{x}^{svr}_{k|k-1} \!\!\!&=&\!\!\! A\hat{x}^{svr}_{k-1|k-1}+Bu_{k-1},\label{eqn:KF-pub-1} \\
P^{svr}_{k|k-1} \!\!\!&=&\!\!\! AP^{svr}_{k-1|k-1}A' + Q,  \\
K^{svr}_{k} \!\!\!&=&\!\!\! P^{svr}_{k|k-1}C'(CP^{svr}_{k|k-1}C' + R)^{-1},  \\
\hat{x}^{svr}_{k|k} \!\!\!&=&\!\!\! \hat{x}^{svr}_{k|k-1} + K^{svr}_{k}(z_{k} - C\hat{x}^{svr}_{k|k-1}),  \\
P^{svr}_{k|k} \!\!\!&=&\!\!\! (I -K^{svr}_{k}C)P^{svr}_{k|k-1},  \label{eqn:KF-pub-5}
\end{eqnarray}
where $\hat{x}^{svr}_{k|k-1}$, $P^{svr}_{k|k-1}$, and $K^{svr}_{k}$ are the server's version of corresponding variables.

Then we propose the privacy metric to measure the quality of privacy preservation at time $k$, the deviation between $\hat{x}^{svr}_{k|k}$ and $\hat{x}_{k|k}$, as
\begin{eqnarray}
\mathcal{Q}_{privacy}^k  \!\!\!&\triangleq&\!\!\! \bm{E}\Big[(\hat{x}^{svr}_{k|k}-\hat{x}_{k|k})(\hat{x}^{svr}_{k|k}-\hat{x}_{k|k})'|\bm{Y}_k\Big], 
\end{eqnarray}
and the average quality of privacy preservation over the entire process as
\begin{eqnarray}\label{eqn:Q_privacy}
\mathcal{Q}_{privacy}\triangleq {1\over T}\sum_{k=0}^{T-1}\mathcal{Q}_{privacy}^k.
\end{eqnarray}
They indicate the deviation of the knowledge of the estimates gained by the server from the true ones.

One may notice that the proposed metric has a quadratic form similar to the error covariance associated with an estimate. Essentially, the privacy problem considered in this paper can be seen as an \emph{``anti-estimation"} problem. 
%Hence, this metric is proper to indicate the privacy level for this type of problems.

%----------remark------------
\begin{remark} \itshape
Compared with the existing privacy metrics, the proposed one has its advantages. 
It is able to apply to a system with only one single user, while differential privacy can be only used in centralized systems with multiple users.
Compared with information-entropy based metrics, % used in the information-theoretic approach, 
the proposed metric works for the state estimate, and is more compatible with the common notions in the study of control systems.
%### CONSIDER TO REVISE

\end{remark}
%----------end remark------------

%-----------------------------------------------------------------------------------------------------------------------
\subsection{Problems to Study}

Although the scheme of noise injection is frequently seen in the studies of privacy preservation, in our considered model, the methodology of local design of the privacy scheme in the background of the  system's closed-loop dynamics makes the analysis seem to be extremely difficult.
By the local design, the injected noise creates a deviation in the server's knowledge, and will further create a deviation in the control input from the original one. The effect of the input deviation will accumulate through the closed-loop as time evolves (presented in Subsection \ref{sec:LQG2-control-performance}). 
This feature enhances the difficulty of analysis.
Moreover, one cannot even affirm the stability of the control system affected by the accumulation of deviations in control inputs.

In this paper, we consider to work on the following problems. 

\begin{enumerate}
\item
The analysis of the quality of privacy preservation $\mathcal{Q}_{privacy}$.

\item
The performance loss in the LQG control as the sacrifice for the privacy preservation.

\item
The trade-off between the privacy preservation and LQG control performances.

\end{enumerate}

%%%%%%%%%%%%%%%%%%%%%%%%%%%%%%%%%%%%%%%%%%%%%
\section{Performance Analysis}

In this section, we present the main results of this paper. We first analyze the performance of privacy preservation of the proposed scheme, and then analyze the performance loss in the LQG control as the sacrifice of service quality. Finally, we propose an optimization problem to study the trade-off between the privacy preservation and the LQG control.

%-----------------------------------------------------------------------------------------------------------------------
\subsection{Privacy Performance}

For the performance of privacy preservation by the proposed privacy scheme, we have the following result.

%--------theorem-----------
\begin{theorem} \label{thm:LQG2-PrivacyPerformance}
It holds that
\begin{eqnarray}
\hat{x}^{svr}_{k|k}\!-\!\hat{x}_{k|k}=(I\!-\!K_k C)A(\hat{x}^{svr}_{k-1|k-1}\!-\!\hat{x}_{k-1|k-1}) \!+\!K_k\delta_k
\end{eqnarray}
and
\begin{eqnarray} \label{eqn:Q-privacy-recursion}
\mathcal{Q}_{privacy}^k=(I\!-\!K_k C)A\mathcal{Q}_{privacy}^{k-1}A'(I\!-\!K_k C)' \!+\!K_k\Sigma_\delta K_k',
\end{eqnarray}
where $K_k$ is the Kalman gain defined in eqn. (\ref{eqn:kalman-gain}). The recursion starts with
$\hat{x}^{svr}_{0|0}-\hat{x}_{0|0}=0$ and $\mathcal{Q}_{privacy}^0=0$.
\end{theorem}
\begin{proof}
In appendix.
\end{proof}
~
%------end theorem--------

We can see that the privacy quality $\mathcal{Q}_{privacy}^k$ is determined by $\Sigma_\delta$. Moreover, we find that a lower bound of $\mathcal{Q}_{privacy}^k$ exists.

%2) when considering the estimation side, the deviations in $u_k$ does not affect the estimation performance. Only the replacement of $y_k$ by $z_k$ in $\hat{x}^{svr}_{k|k}$ matters.

%--------theorem-----------
\begin{theorem} \label{thm:Q-privacy-bound}
A matrix sequence $\{M_k\}$, $M_k\in \mathbb{S}_{+}^{n}$,  satisfies that $M_0=0$, $M_1=\mathcal{Q}_{privacy}^1$, and when $k\geq 2$,
\begin{eqnarray}
M^-_k\!\!\!&=&\!\!\! AM_{k-1}A',  \\
M_k\!\!\!&=&\!\!\! M^-_k-M^-_kC'(CM^-_kC'+\Sigma_\delta)^{-1}CM^-_k.
\end{eqnarray}
Then
\begin{eqnarray*}
\mathcal{Q}_{privacy}^k\geq M_k,~~~\forall k.
\end{eqnarray*}

\end{theorem}
\begin{proof}
Firstly, we have $\mathcal{Q}_{privacy}^0=M_0$ and $\mathcal{Q}_{privacy}^1=M_1$.  Since $M_1\neq 0$, the sequence $\{M_k\}$ will not be all 0. Assume that $\mathcal{Q}_{privacy}^{k-1}\geq M_{k-1}$ holds.
Consider the operator
\begin{eqnarray*}
\phi(X,\mathit\Gamma) \triangleq (I-\mathit\Gamma C)X(I-\mathit\Gamma C)'+\mathit\Gamma \Sigma_\delta \mathit\Gamma'.
\end{eqnarray*}
Then $\mathcal{Q}_{privacy}^k=\phi(A\mathcal{Q}_{privacy}^{k-1}A',K_k)$. Let 
\begin{eqnarray*}
\mathit\Gamma^M_k\triangleq M^-_kC'(CM^-_kC'+\Sigma_\delta)^{-1},
\end{eqnarray*}
and we have $M_k=\phi(M^-_{k-1},\mathit\Gamma^M_k)$. Furthermore, it is straightforward to verify that 
\begin{eqnarray*}
\mathit\Gamma^M_k=\argmin_{\mathit\Gamma} \phi(M^-_{k-1},\mathit\Gamma).
\end{eqnarray*}
Since $\mathcal{Q}_{privacy}^{k-1}\geq M_{k-1}$, we have $A\mathcal{Q}_{privacy}^{k-1}A' \geq M^-_{k-1}$.
Notice that $\phi(X,\mathit\Gamma)$ is linear in $X$. Then we have 
\begin{eqnarray*}
\phi(A\mathcal{Q}_{privacy}^{k-1}A',K_k)\geq \phi(M^-_{k-1},K_k) \geq \phi(M^-_{k-1},\mathit\Gamma^M_k),
\end{eqnarray*}
which is $\mathcal{Q}_{privacy}^k\geq M_k$.
\end{proof}
~
%------end theorem--------

As an estimate of $x_k$, $\hat{x}^{svr}_{k|k}$ also has its associated error covariance.
From the proof of Theorem \ref{thm:LQG2-PrivacyPerformance}, we find that $K^{svr}_k=K_k$ and $P^{svr}_{k|k}=P_{k|k}$, which shows that $P^{svr}_{k|k}$ should not be the real error covariance associated with $\hat{x}^{svr}_{k|k}$.
We denote the real error covariance associated with $\hat{x}^{svr}_{k|k}$ as $P^{svr.real}_{k|k}$.

%---------lemma-----------
\begin{lemma}
It holds that
\begin{eqnarray*}
P^{svr.real}_{k|k}= P_{k|k}+\mathcal{Q}_{privacy}^k.
\end{eqnarray*}
\end{lemma}
\begin{proof}
We have
\begin{eqnarray*}
&&P^{svr.real}_{k|k}\\
\!\!\!&=&\!\!\! \bm{E}\Big[(x_k-\hat{x}^{svr}_{k|k})(x_k-\hat{x}^{svr}_{k|k})'|\bm{Y}_k\Big]\\
\!\!\!&=&\!\!\! \bm{E}\Big[(x_k\!-\!\hat{x}_{k|k}\!+\!\hat{x}_{k|k}\!-\!\hat{x}^{svr}_{k|k})(x_k\!-\!\hat{x}_{k|k} \!+\!\hat{x}_{k|k}\!-\!\hat{x}^{svr}_{k|k})'|\bm{Y}_k\Big]\\
\!\!\!&=&\!\!\! P_{k|k}+\mathcal{Q}_{privacy}^k+\bm{E}\Big[(x_k-\hat{x}_{k|k})(\hat{x}_{k|k} -\hat{x}^{svr}_{k|k})'|\bm{Y}_k\Big]\\
&&\!\!\!+\bm{E}\Big[(\hat{x}_{k|k} -\hat{x}^{svr}_{k|k})(x_k-\hat{x}_{k|k})'|\bm{Y}_k\Big].
%\!\!\!&=&\!\!\! \color{red} P_{k|k}+\mathcal{Q}_{privacy}^k~?
\end{eqnarray*}
Notice that $\hat{x}_{k|k}$ is the linear combination of $\bm{Y}_k$, $\hat{x}^{svr}_{k|k}$ is the linear combination of $\bm{Y}_k$ and $\{\delta_1,...,\delta_k\}$, and $x_k-\hat{x}_{k|k}$ is independent of $\bm{Y}_k$ and $\{\delta_1,...,\delta_k\}$. Then the two expectation terms in the last equality equals 0.
\end{proof}
%------end lemma--------

%\textbf{Infinite-time case:} 

%---------remark----------
\begin{remark}\itshape
The local design scheme misleads the server's knowledge so that $\hat{x}^{svr}_{k|k}$ is deviated from the true one $\hat{x}_{k|k}$.
From another perspective, the deviation is equivalent to that the server's seized estimate $\hat{x}^{svr}_{k|k}$ is correct but has a larger error covariance $P^{svr.real}_{k|k}$, which is unknown to the server.

\end{remark}
%---------end remark----------

%---------remark----------
\begin{remark}[A Doubtful Server]\itshape
In our paper, the server entirely trusts the data provided by the user and has no data analysis mechanism. 
If the server is ``doubtful'' and has an additional detector to analyze the received measurement sequence $z_k$, it will find that the covariance of $z_k$ is not as it supposed to be the one in the provided parameters.
The worst case is that the server successfully recovers $\Sigma_\delta$, the covariance of the injected noise, and it also makes the correct guess that it is due to noise injection in the measurement sequence (notice that even the server finds the mismatch of the measurement covariance, there are still many possibilities how the user makes the situation).
Then the server could adjust the knowledge of the measurement noise covariance $R$ by $R^{sl}=R+\Sigma_\delta$, and develop its own filtering process to estimate state $x_k$:
\begin{eqnarray*}
\hat{x}^{sl}_{k|k-1} \!\!\!&=&\!\!\! A\hat{x}^{sl}_{k-1|k-1}+Bu_{k-1}, \\
P^{sl}_{k|k-1} \!\!\!&=&\!\!\! AP^{sl}_{k-1|k-1}A' + Q,  \\
K^{sl}_{k} \!\!\!&=&\!\!\! P^{sl}_{k|k-1}C'(CP^{sl}_{k|k-1}C' + R^{sl})^{-1},  \\
\hat{x}^{sl}_{k|k} \!\!\!&=&\!\!\! \hat{x}^{sl}_{k|k-1} + K^{sl}_{k}(z_{k} - C\hat{x}^{sl}_{k|k-1}),  \\
P^{sl}_{k|k} \!\!\!&=&\!\!\! (I -K^{sl}_{k}C)P^{sl}_{k|k-1}, 
\end{eqnarray*}
where the superscript ``sl'' represents the corresponding variables are the server's local version.
This case is equivalent to that the user shares its privacy parameters to the server, as in many previous studies \cite{Tanaka2017ACC, Hale2018ACC, Yazdani2023TAC}. Nevertheless, the knowledge of the server is still corrupted by the noise and hence the user still preserves certain level of privacy. %Moreover, possible improvements for the user subject to this issue is discussed in the section of future work at the end of the paper.

\end{remark}
%---------end remark----------

%---------remark----------
\begin{remark}\itshape

A feasible analysis on the state trajectory which can be done by the server is to smooth past estimates after the process ends: the server could compute $\hat{x}^{svr}_{k|T}$, where $T$ is the time horizon of the process. Two problems can be visited: 
1) what is smoothed result $\hat{x}^{svr}_{k|T}$ like and how it is affected by the future information where noise injection also exists,
2) the comparison of $\hat{x}^{svr}_{k|T}$ with both $\hat{x}_{k|k}$ and $\hat{x}_{k|T}$, respectively.

\end{remark}
%---------end remark----------
%-----------------------------------------------------------------------------------------------------------------------
\subsection{LQG Control Performance} \label{sec:LQG2-control-performance}

The server computes the control law based on the knowledge of the user's state estimate. 
Under the privacy scheme, the server computes the control input as
\begin{eqnarray}
 u_k=L_k\hat{x}^{svr}_{k|k}.
\end{eqnarray}
Define
\begin{eqnarray*}
\hat{d}_k\triangleq\hat{x}^{svr}_{k|k}-\hat{x}_{k|k}.
\end{eqnarray*}
Hence, the true dynamics now is as follows:
\begin{eqnarray} \label{eqn:LQG2-system-after-scheme}
\aligned
x_{k+1}&=A x_k+B u_k+w_k, \\
u_k&= L_k(\hat{x}_{k|k}+\hat{d}_k).
%y_k\!\!\!&=&\!\!\!Cx_k+v_k.
\endaligned
\end{eqnarray}

In eqn. (\ref{eqn:LQG2-system-after-scheme}), we can see the caused effect by the privacy scheme. 
The injected noise also creates a deviation in the control input which will accumulate through the closed loop, and hence makes the system performance complex to analyze.
The control law set $\{u_k\}$ is non-optimal now; therefore, the resulting state trajectory $\{x_k\}$ is also non-optimal. 
Assume that the noise sequences $\{w_k\}$ and $\{v_k\}$ are given. 
Let the state trajectory under the optimal control law $\{u^\star_k\}$ be $\{x^\star_k\}$ and its  estimate be $\{\hat{x}^\star_{k|k}\}$ in the system without privacy scheme. 
%Since the control inputs are non-optimal now, the resulted states and hence the associated estimates also deviate from the optimal values. 
We can see that $\{u_k\}$, $\{x_k\}$ and $\{\hat{x}_{k|k}\}$ are all deviated from $\{u^\star_k\}$, $\{x^\star_k\}$ and $\{\hat{x}^\star_{k|k}\}$, respectively, and hence the LQG objective is also affected.

Then we consider the sacrifice in the LQG performance caused by the privacy scheme.
Without misunderstanding, we still use the notation $\mathcal{O}_{0:T}$ to denote the objective under the privacy scheme, which also includes the effect of the uncertainty of $\delta_k$ now.
%We can compare $\{x_k\}$ with $\{x^\star_k\}$ and compare $\{\hat{x}_{k|k}\}$ with $\{\hat{x}^\star_{k|k}\}$, and check their deviations.
%claim the optimal estimation only associate with optimal controller
Define the metric of the average performance loss in LQG control as
\begin{eqnarray}
\mathcal{Q}_{LQG}\triangleq{1\over T}(\mathcal{O}_{0:T}-\mathcal{O}^\star_{0:T}).
\end{eqnarray}

For the proposed privacy scheme, one may have concern that the deviations in control inputs will accumulate through the closed loop and the system dynamics becomes complex to analyze.
In the following theorem, we are able to show that the LQG performance still has a neat closed-form expression.

%--------theorem-----------
\begin{theorem} \label{thm:LQG2-obj}
Under the privacy scheme, it holds that
\begin{eqnarray*}
\mathcal{O}_{0:T}\!\!\!&=&\!\!\!\bm{E}(x_0'S_0x_0)+\sum_{k=0}^{T-1}\mathrm{tr}(S_{k+1}Q) + \sum_{k=0}^{T-1}\mathrm{tr}(\mathit\Phi_{k}P_k)\\
&&\!\!\!+\sum_{k=0}^{T-1}\mathrm{tr}\Big(\mathit\Phi_{k}\mathcal{Q}_{privacy}^k \Big),
\end{eqnarray*}
where
\begin{eqnarray*}
\mathit\Phi_k=A'S_{k+1}B(B'S_{k+1}B+U)^{-1}B'S_{k+1}A.
\end{eqnarray*}
\end{theorem}
\begin{proof}
In appendix.
\end{proof}
%------end theorem-------

We find that the employment of the privacy scheme produce an additional item in the objective value compared with eqn. (\ref{eqn:LQG-obj-opt}).
This item is related to the privacy metric $\mathcal{Q}_{privacy}^k$, which mathematically says that the LQG performance loss is caused by the privacy preservation.
Hence, the performance loss of the LQG control is given as
\begin{eqnarray} \label{eqn:Q_LQG}
\mathcal{Q}_{LQG} ={1\over T}\sum_{k=0}^{T-1}\mathrm{tr}\Big(\mathit\Phi_{k}\mathcal{Q}_{privacy}^k \Big).
\end{eqnarray}

Furthermore, by analyzing the term $\mathcal{Q}_{LQG}$, we find that it is bounded as time approaches infinity when the injected noise has bounded covariance, which means that in the long run the privacy scheme will not cause the instability of the control performance. This result is presented as follows.

%--------theorem-----------
\begin{theorem} \label{thm:boundness-Q-LQG}
When $\Sigma_\delta<+\infty$, it holds that
\begin{eqnarray} \label{eqn:Q_LQG}
\lim_{T\rightarrow\infty}{1\over T}\sum_{k=0}^{T-1}\mathrm{tr}\Big(\mathit\Phi_{k}\mathcal{Q}_{privacy}^k \Big)<+\infty.
\end{eqnarray}

\end{theorem}
%~~~~~~~~~~
\begin{proof}
Firstly, in the previous section, we have assumed that $(A, B)$ controllable and $(\sqrt{W}, A)$ is detectable.
Based on the two conditions, when $T\rightarrow\infty$, for a given $k$, $S_k$ defined in eqn. (\ref{eqn:LQG-S-recursion}) converges to a steady value \cite{Betsekas1995}. We define
\begin{eqnarray} \label{eqn:steady-S}
S=\lim_{T\rightarrow\infty} S_k.
\end{eqnarray}
Then $S$ is the unique solution satisfying
\begin{eqnarray}
%L &=& (U+B'SB)^{-1}B'SA,  \label{eqn:L-time-invariant}\\
S=A'S A+W-A'S B(U+B'S B)^{-1}B'S A.  \label{eqn:S-time-invariant}
\end{eqnarray}
Since for any given $k$, the limit in eqn. (\ref{eqn:steady-S}) exists, we can also write $\lim\limits_{k\rightarrow\infty} S_k = S$.
Accordingly, $\mathit\Phi_{k}$ also approaches to a steady value:
\begin{eqnarray*}
\lim_{k\rightarrow\infty} \mathit\Phi_{k}=A'SB(B'SB+U)^{-1}B'SA.
\end{eqnarray*}
Secondly, since $K_k$ is the Kalman gain, $(I-K_kC)A$ is stable \cite{AndersonMoore1979}, i.e., the spectrum radius $\rho\Big((I-K_kC)A\Big)<1$.
Moreover, based on the assumed conditions $(A, \sqrt{Q})$ is stabilizable and $(C, A)$ is detectable, when $T\rightarrow\infty$, $K_k$ reaches a steady value according to the property of Kalman filtering \cite{AndersonMoore1979}:
\begin{eqnarray*}
K=\lim_{k\rightarrow\infty} K_k.
\end{eqnarray*}
Let $k\rightarrow\infty$, and eqn. (\ref{eqn:Q-privacy-recursion}) becomes 
\begin{eqnarray*} 
\lim_{k\rightarrow\infty}\!\!\mathcal{Q}_{privacy}^k\!=\!(I\!-\!K C)A\!\lim_{k\rightarrow\infty}\!\!\mathcal{Q}_{privacy}^{k-1}A'(I\!-\!K C)' \!+\!K\Sigma_\delta K'.
\end{eqnarray*}
The existence of $\lim\limits_{k\rightarrow\infty}\mathcal{Q}_{privacy}^k$ needs to be verified.
Consider the Lyapunov equation
\begin{eqnarray} \label{eqn:Q-Lyapunov-eqn}
X=(I-K C)AXA'(I-K C)' +K\Sigma_\delta K'.
\end{eqnarray}
Since $(I-K_kC)A$ is stable, we also have that $(I-KC)A$ is stable. 
When $\Sigma_\delta$ is bounded, the existence of a unique solution to (\ref{eqn:Q-Lyapunov-eqn}) is guaranteed \cite{Kailath2000}.
Then the limit $\lim\limits_{k\rightarrow\infty}\mathcal{Q}_{privacy}^k$ exists.
Moreover, it can be explicitly given as
\begin{eqnarray*}
\lim\limits_{k\rightarrow\infty}\mathcal{Q}_{privacy}^k=\sum_{k=0}^{\infty}\Big[(I-K_kC)A\Big]^k K\Sigma_\delta K' \Big[A'(I-K_kC)'\Big]^k,
\end{eqnarray*}
where we can see that it is linear with $\Sigma_\delta$.
%Hence, when $\Sigma_\delta$ is bounded,   $\lim\limits_{k\rightarrow\infty}\mathcal{Q}_{privacy}^k$ is also bounded.

Since both $\lim\limits_{T\rightarrow\infty} \mathit\Phi_{k}$ and $\lim\limits_{k\rightarrow\infty} \mathcal{Q}_{privacy}^k$  are bounded, we have
\begin{eqnarray*}
\lim_{T\rightarrow\infty}{1\over T}\sum_{k=0}^{T-1}\mathrm{tr}\Big(\mathit\Phi_{k}\mathcal{Q}_{privacy}^k \Big)=\lim_{k\rightarrow\infty} \mathrm{tr}\Big(\mathit\Phi_{k}\mathcal{Q}_{privacy}^k\Big)<+\infty.
\end{eqnarray*}
\end{proof}
%------end theorem-------

%--------remark-----------
\begin{remark}\itshape
From the proof of Theorem \ref{thm:boundness-Q-LQG}, we can see that the asymptotic stability of the system under the proposed privacy scheme is essentially guaranteed by the effects of both the Kalman filter and the original optimal LQG control.
The form of privacy scheme, which is the white noise injection, also matters, as the influence of the noise happens to be refrained by the Kalman filter and the feedback control. 
If the injection scheme is sophisticated specified rather than i.i.d. noise injection, the stability of the system needs to be further verified.

\end{remark}
%--------end remark-----------

Moreover, we can extend the result to a system model with a general stable state feedback gain. We present the following extension.

%--------corollary----------
\begin{corollary}
Consider a reduced original system that the user directly sends its state $x_k$ to the server and the server computes a general linear state feedback control policy:
\begin{eqnarray} \label{eqn:LQG1}
u_k&= F_k x_k,
\end{eqnarray}
where the feedback gain $F_k$ satisfies that $A+BF_k$ is stable.
In this scenario, the user's privacy scheme is to send $z_k=x_k+\delta_k$ to the server.
Then the objective eqn. (\ref{eqn:obj2}) under this control policy and privacy scheme follows
\begin{eqnarray*}
\mathcal{O}_{0:T} &=& x_0'\Omega_0x_0+\sum_{k=0}^{T-1}\mathrm{tr}(\Omega_{k+1}Q) \\
&&+\sum_{k=0}^{T-1}\mathrm{tr}\Big(F_k'(U+B'\Omega_{k+1}B)F_k\Sigma_\delta \Big),
\end{eqnarray*}
where $\Omega_{k}$ follows
\begin{eqnarray*} 
\Omega_T\!\!\!\!&=&\!\!\!\! W, \label{eqn:LQG-Omega-terminal}\\
\Omega_k\!\!\!\!&=&\!\!\!\!W+F_k'UF_k+(A+BF_k)'\Omega_{k+1}(A+BF_k). \label{eqn:LQG-Omega-recursion}
\end{eqnarray*}

\end{corollary}
\begin{proof}
The proof can be similarly obtained following the procedures of the proof of Theorem \ref{thm:LQG2-obj}.
\end{proof}
%------end corollary-------

%-----------------------------------------------------------------------------------------------------------------------
\subsection{Optimization Problem}

After the preious analyisis, we can see that the privacy preservation is gained at the price of the LQG control performance. 
To study the trade-off, we propose the following optimization problem: to maximize the privacy metric $\mathcal{Q}_{privacy}$ when the loss of LQG control performance $\mathcal{Q}_{LQG}$ is required to be under a given level $\alpha$. 
This problem is formulated as follows.

%--------problem-----------
\begin{problem}\label{prblm:optimization2} 
\begin{eqnarray*}
&  \max\limits_{\Sigma_\delta} \!\!\! &\mathrm{tr}(\mathcal{Q}_{privacy}) \\
 &  \mathrm{s.t.} \!\!\!&   \Sigma_\delta\geq 0, \\
& & \mathcal{Q}_{LQG}\leq \alpha,  \\
& & \mathcal{Q}_{LQG} ={1\over T}\sum_{k=0}^{T-1}\mathrm{tr}\Big(\mathit\Phi_{k}\mathcal{Q}_{privacy}^k \Big),\\
& & \mathcal{Q}_{privacy}^{k}\!=\!(I\!-\!K_k C)A\mathcal{Q}_{privacy}^{k-1}A'(I\!-\!K_k C)'
 \!+\! K_k\Sigma_\delta K_k', \\
&&~~~~~~~~~~~~~~~~~~~~~~~~~~~~~~~~~~~~~~~~~~~~ k=1,2,...,T.
%&                         & ~~~
\end{eqnarray*}
\end{problem}
%--------end problem-----------

This problem turns out to be a semidefinite programming (SDP) problem, which can be solved efficiently by numerical methods for convex optimization \cite{Boyd2004}.
For the theoretical analysis, this problem is taken as solved at this stage. One algorithm for solving the problem using MATLAB software CVX \cite{CVX} is stated as follows.

%--------algorithm-----------
\begin{algorithm}
	\caption{Numerical Solution to Problem \ref{prblm:optimization2} with CVX }
	\begin{algorithmic}[1]
%			\setstretch{1.25}
           \STATE  \textbf{Paremeters:} system order $n$, measurement order $m$, time horizon $T$, the required control performance level $\alpha$.
		\STATE	\textbf{Input:} matrices $A$, $B$, $C$, $Q$, $R$, $W$, and $U$.  \\
~

		\STATE  Compute $S_k$, $L_k$, and $\mathit\Phi_k$ backwards from time $T$ to 0 according to eqn. (\ref{eqn:LQG-S-terminal})-(\ref{eqn:LQG-gain}).
          \STATE Compute $K_k$ from time 0 to $T$ according to eqn. (\ref{eqn:error-covariance-time-update}), (\ref{eqn:kalman-gain}), and (\ref{eqn:error-covariance-measurement-update}). \\

~
          \STATE \textbf{Begin CVX: solving SDP problem}	
          \STATE ~~~Declare variables: $\mathcal{Q}^k_{privacy}$, $\Sigma_{\delta}$.
		\STATE ~~~Objective: to maximize $\mathrm{tr}(\mathcal{Q}_{privacy})$.
           \STATE ~~~Constraints:
		\STATE ~~~~~~(1) $\mathcal{Q}_{privacy}$ is constrained by eqn. (\ref{eqn:Q_privacy});
		\STATE ~~~~~~(2) $\mathcal{Q}_{LQG}$ is constrained by eqn. (\ref{eqn:Q_LQG});
 		\STATE ~~~~~~(3) let $\mathcal{Q}_{LQG}$ be less than or equal to $\alpha$;
		\STATE ~~~~~~(4) vectorize both sides of eqn. (\ref{eqn:Q-privacy-recursion}) for $i=0, ..., T$, \\~~~~~~~~~~and pile the outcomes into one equality.
  
           \STATE  \textbf{End CVX}
	\end{algorithmic}
\end{algorithm}
%--------end algorithm-----------

%-----------------------------------------------------------------------------------------------------------------------
\subsection{Special Case Study: Infinite-Time Horizon}

Particularly, we consider a special case that the time horizon $T$ becomes infinite.
Based on the analysis in the proof of Theorem \ref{thm:boundness-Q-LQG}, the system approaches to steady behaviors when time goes infinity. 
For the estimator, the Kalman gain converges: $\lim\limits_{k\rightarrow\infty}K_k = K$.
For the controllor, when $T\rightarrow\infty$, $\lim\limits_{k\rightarrow\infty}S_k=S$, where $S$ is given by eqn. (\ref{eqn:S-time-invariant}). Meanwhile, the associated LQG gain $L_k$ also becomes time-invariant.

According to the proof of Theorem \ref{thm:boundness-Q-LQG}, the steady value $\lim\limits_{k\rightarrow\infty}\mathcal{Q}_{privacy}^k$ exists.
We have the average deviation covariance in this infinite-time case as follows:
\begin{eqnarray*}
\mathcal{Q}_{privacy}=\lim_{T\rightarrow\infty}{1\over T}\sum_{k=0}^{T-1} \mathcal{Q}_{privacy}^k.
\end{eqnarray*}
Then we find that
\begin{eqnarray*}
\mathcal{Q}_{privacy}=\lim_{k\rightarrow\infty}\mathcal{Q}_{privacy}^k,
\end{eqnarray*}
i.e., $\mathcal{Q}_{privacy}$ coincides with the steady state of $\mathcal{Q}_{privacy}^k$.
Then we have
\begin{eqnarray}  \label{eqn:Q-privacy-infinite}
\mathcal{Q}_{privacy}=(I\!-\!K C)A\mathcal{Q}_{privacy}A'(I\!-\!K C)' \!+\!K\Sigma_\delta K'.
\end{eqnarray}

We also present the lower bound of $\mathcal{Q}_{privacy}$ in this case.

\begin{corollary}
Let $M$ and $M^-$ be the solutions to
\begin{eqnarray*}
M^-\!\!\!&=&\!\!\! AMA',  \\
M\!\!\!&=&\!\!\! M^--M^-C'(CM^-C'+\Sigma_\delta)^{-1}CM^-.
\end{eqnarray*}
Then 
\begin{eqnarray*}
\mathcal{Q}_{privacy}\geq M.
\end{eqnarray*}
\end{corollary}
\begin{proof}
We can find $M=\lim\limits_{k\rightarrow\infty} M_k$, where $M_k$ is defined in Theorem \ref{thm:Q-privacy-bound}.  Since $\mathcal{Q}_{privacy}^k\geq M_k$ for all $k$, then we have the result.
\end{proof}

In the infinite-time case, the formulation is reduced accordingly.
The loss of LQG performance is simplified as follows:
\begin{eqnarray*}
\mathcal{Q}_{LQG}=\lim_{T\rightarrow\infty}{1\over T}\sum_{k=0}^{T-1} \mathrm{tr}\Big(\mathit\Phi_{k}\mathcal{Q}_{privacy}^k \Big)= \mathrm{tr}(\mathit\Phi \mathcal{Q}_{privacy}),
\end{eqnarray*}
where $\mathit\Phi =\lim\limits_{k\rightarrow\infty}\mathit\Phi_{k}=A'SB(U\!+\!B'SB)^{-1}B'SA$.
%Since $\mathcal{Q}_{LQG}$ is finite, the stability of the system dynamics is guaranteed.

%{\color{red}Problem \ref{prblm:optimization2} }

Then we study the the following optimization problem:
\begin{problem}  \label{problem:LQG2-infinite}
\begin{eqnarray*}
 &   \max\limits_{\Sigma_\delta}  &\mathrm{tr}(\mathcal{Q}_{privacy})\\
 &  \mathrm{s.t.}   &  \Sigma_\delta\geq 0, \\ 
&& \mathrm{tr}(\mathit\Phi\mathcal{Q}_{privacy})\leq \alpha,\\
&& \mathcal{Q}_{LQG} = \mathrm{tr}(\mathit\Phi \mathcal{Q}_{privacy}), \\
&&\mathcal{Q}_{privacy}=(I\!-\!K C)A\mathcal{Q}_{privacy}A'(I\!-\!KC)' +K\Sigma_\delta K'.
\end{eqnarray*}
\end{problem}
This problem is also a semidefinite positive problem and can be solved efficiently by classical numerical methods \cite{Boyd2004}. Examples and simulations are presented in Section \ref{sec:examples}.

\begin{lemma}\label{lemma:LQG2-scalar}
In the scalar system, the solution to Problem \ref{problem:LQG2-infinite} is given by
\begin{eqnarray*}
\Sigma_\delta^\star={ \alpha \over  \mathit\Phi K^2}\Big[1-A^2(1-KC)^2\Big]
\end{eqnarray*}
and 
\begin{eqnarray*}
\mathcal{Q}^\star_{privacy}={ \alpha \over  \mathit\Phi}.
\end{eqnarray*}
\end{lemma}
\begin{proof}
The result is solved by direct calculation.
\end{proof}

In the scalar system, we can see the proportional relationship between $\mathcal{Q}^\star_{privacy}$, $\Sigma_\delta^\star$, and $\alpha$.

%%%%%%%%%%%%%%%%%%%%%%%%%%%%%%%%%%%%%%%%%%%%%
\section{Perfect Privacy Preservation without Performance Loss} \label{sec:perfect-control}

To the user, the best situation is that the server has deviated knowledge of the states while the user is still able to obtain optimal control inputs. We provide simple results on this direction in this section.

%-----------------------------------------------------------------------------------------------------------------------
\subsection{System with Accessible State Information}

%Using the Kernel Space

We first consider a simple scenario. 
Assume that the user is able to directly have the value of $x_k$ and then does not need the sensor.
This is a simplified version of the scenario we consider in this paper.
Then the privacy scheme becomes 
\begin{eqnarray}\label{eqn:privacy-scheme-func}
z_k=x_k+\delta_k.
\end{eqnarray}
In this case, the server will take $z_k$ as the user's state $x_k$.
We denote the trajectory the server sees as $\{x^{svr}_k\}$.
According to the privacy scheme, $x^{svr}_k$ is identical to $z_k$: 
\begin{eqnarray} \label{eqn:LQG1-scheme}
x^{svr}_k=x_k+\delta_k.
\end{eqnarray}
Similarly, the server computes the control input according to
\begin{eqnarray}
u_k=L_k x^{svr}_k.
\end{eqnarray}
Denote the LQG objective in this scenario as $\mathcal{I}_{0:T}$ and the optimal objective as $\mathcal{I}_{0:T}^\star$.

Usually, the added noise $\delta_k$ in the transmitted signal also causes a corresponding additional deviation in the generated control input from the optimal one. 
However, this added $\delta_k$ does not always cause a deviation in the control input. 
As one example, 
if $\delta_k$ lies in $\mathscr{N}(L_k)$, the kernel space of $L_k$, i.e., $L_k\delta_k=0$, we have 
\begin{eqnarray*}
%u_0\!\!\!&=&\!\!\!  L_0x^{svr}_0=L_0(x_0+\delta_0)=L_0x_0=u^\star_0, \\
%&\vdots& \\
u_k\!\!\!&=&\!\!\! L_kx^{svr}_k=L_k(x_k+\delta_k)=L_kx_k,%=u^\star_k.
\end{eqnarray*}
i.e., the resulted $u_k$ is the optimal control input for the current state $x_k$.
We follow this idea and obtain the following result.
\begin{theorem}
Given an arbitrary time $k$, for the variable $\delta_k$ in eqn. (\ref{eqn:LQG1-scheme}), if 
\begin{enumerate}
\item
$\delta_k\in\mathscr{N}(A)$ when $A$ is not of full rank, 

\item
or $S_{k+1}A\delta_k\in\mathscr{N}(B')$,
\end{enumerate}
it holds that $u_k=L_kx_k$, which is the optimal control input at time $k$.
If for each time $k$, within 1) and 2), there always exists one condition which is satisfied, then $u_k=u^\star_k$ and $x_k=x^\star_k$ for all $k$. Moreover, $\mathcal{I}_{0:T}=\mathcal{I}^\star_{0:T}$.

%the server will still generate optimal control input under the proposed privacy scheme. 
\end{theorem}
\begin{proof}
Notice that $L_k=-(B'S_{k+1}B+U)^{-1}B'S_{k+1}A$ and $B'S_{k+1}B+U$ has full rank.
Either condition 1) or 2) leads to that $L_k\delta_k=0$.
\end{proof}

In this methodology, the variable $\delta_k$ is unnecessary to be random. This methodology may create arbitrary deviation in the server's knowledge of $x_k$ while still maintaining optimal control performance.

However, the problem is, this methodology needs the user to do complex computation before the process. Usually, the reason for a user to employ a server is that it lacks necessary computation capability. But maybe the user could employ another server to compute this particular set of $\delta_k$ for it.

%-----------------------------------------------------------------------------------------------------------------------
\subsection{System with Inaccessible State Information}

In this scenario, which is mainly studied in this paper, it becomes more difficult to find the perfect preservation scheme than in the previous scenario. We present a result for the system with the matrix $A$ without full rank.

\begin{theorem}
If $A$ is not of full rank, 
for all time $k$, if the variable $\delta_k$ in eqn. (\ref{eqn:LQG2-scheme}) satisfies that $K_k\delta_k\neq 0$ and $AK_k\delta_k=0$, 
%\begin{enumerate}
%\item
%$K_k\delta_k\in\mathscr{N}(A)$ , 
%
%%\item
%%or $S_{k+1}AK_k\delta_k\in\mathscr{N}(B')$,
%
%\item 
%and $(CP_{k|k-1}C' + R)^{-1}\delta_k\notin \mathscr{N}(C')$,
%
%\end{enumerate}
then $u_k=u^\star_k$, $\hat{x}_{k|k}=\hat{x}^\star_{k|k}$, and $\hat{x}^{svr}_{k|k}\neq \hat{x}_{k|k}$ for all $k$. Moreover, $\mathcal{O}_{0:T}=\mathcal{O}^\star_{0:T}$.
%the server will still generate optimal control input under the proposed privacy scheme. 
\end{theorem}
\begin{proof}
We have $\hat{x}^{svr}_{0|0}=\hat{x}_{0|0}=\hat{x}^\star_{0|0}$ initially. Then %$u_0=L_0\hat{x}^{svr}_{0|0}$
\begin{eqnarray*}
u_0=L_0\hat{x}^{svr}_{0|0}=L_0\hat{x}^\star_{0|0}=u^\star_0.
\end{eqnarray*}
Then we have
\begin{eqnarray*}
\hat{x}^{svr}_{1|0} =A\hat{x}^{svr}_{0|0}+Bu_0=A\hat{x}^\star_{0|0}+Bu^\star_0=\hat{x}^\star_{1|0}.
\end{eqnarray*}
%Assume that $u_0=u^\star_0, \cdots, u_{k-1}=u^\star_{k-1}$ and 
It is simple to find $\hat{x}_{1|1}=\hat{x}^\star_{1|1}$.
%According to condition 2),  we have
%\begin{eqnarray*}
%K_1\delta_1=P_{1|0}C'(CP_{1|0}C' + R)^{-1}\delta_1\neq 0,
%\end{eqnarray*}
Moreover, since $K_1\delta_1\neq 0$,
\begin{eqnarray*}
\hat{x}^{svr}_{1|1} \!\!\!&=&\!\!\! \hat{x}^{svr}_{1|0} \!+\! K^{svr}_{1}(z_{1} \!-\! C\hat{x}^{svr}_{1|0}) \\
 \!\!\!&=&\!\!\!\hat{x}^\star_{1|0} \!+\! K_{1}(y_{1}+\delta_1 \!-\! C\hat{x}^\star_{1|0})\\
 \!\!\!&\neq&\!\!\!\hat{x}^\star_{1|0} \!+\! K_{1}(y_{1} \!-\! C\hat{x}^\star_{1|0})= \hat{x}^\star_{1|1}=\hat{x}_{1|1}.
\end{eqnarray*}
Consider $k=2$, we have 
\begin{eqnarray*}
u_1\!\!\!&=&\!\!\!L_1\hat{x}^{svr}_{1|1}=L_1\Big[\hat{x}^\star_{1|0} \!+\! K_{1}(z_{1} \!-\! C\hat{x}^\star_{1|0})\Big]\\
\!\!\!&=&\!\!\!L_1\Big[\hat{x}^\star_{1|0} \!+\! K_{1}(y_{1} \!-\! C\hat{x}^\star_{1|0})\Big]+L_1K_1\delta_1\\
\!\!\!&=&\!\!\! L_1\hat{x}^\star_{1|1}+L_1K_1\delta_1.
%L_1\hat{x}^\star_{0|0}=u^\star_0.
\end{eqnarray*}
Since $L_1K_1\delta_1=-(B'S_{2}B+U)^{-1}B'S_{2}AK_1\delta_1$ and $AK_1\delta_1=0$, we have $L_1K_1\delta_1=0$.
Then it holds that
\begin{eqnarray*}
u_1=L_1\hat{x}^\star_{1|1}=u^\star_1.
\end{eqnarray*}
Further we have
\begin{eqnarray*}
\hat{x}^{svr}_{2|1}\!\!\!&=&\!\!\!A\hat{x}^{svr}_{1|1}+Bu_1=A\Big[\hat{x}^\star_{1|0} \!+\! K_{1}(z_{1} \!-\! C\hat{x}^\star_{1|0})\Big]+Bu^\star_1\\
\!\!\!&=&\!\!\! A\Big[\hat{x}^\star_{1|0} \!+\! K_{1}(y_{1} \!-\! C\hat{x}^\star_{1|0})\Big]+AK_1\delta_1+Bu^\star_1\\
\!\!\!&=&\!\!\! A\hat{x}^\star_{1|1}+Bu^\star_1=\hat{x}^\star_{2|1}.
%=\hat{x}^\star_{1|0}.
\end{eqnarray*}
We extend the derivation to general $k$ and obtain $u_k=u^\star_k$, $\hat{x}_{k|k}=\hat{x}^\star_{k|k}$, and $\hat{x}^{svr}_{k|k}\neq \hat{x}_{k|k}$ for all $k$. We also see that $\hat{x}^{svr}_{k|k-1}= \hat{x}_{k|k-1}=\hat{x}^\star_{k|k-1}$.
\end{proof}

%\begin{eqnarray*}
%\hat{x}^{svr}_{0|0}\!\!\!&=&\!\!\!\hat{x}_{0|0},\\
%\hat{x}^{svr}_{1|1} \!\!\!&=&\!\!\! \hat{x}^{svr}_{1|0} \!+\! K^{svr}_{1}(z_{1} \!-\! C\hat{x}^{svr}_{1|0}) \\
%\!\!\!&=&\!\!\!\hat{x}^{svr}_{1|0} \!+\! K^{svr}_{1}(y_{1} \!-\! C\hat{x}^{svr}_{1|0})=\hat{x}_{1|1},\\
%&\vdots& \\
%\hat{x}^{svr}_{k|k}\!\!\!&=&\!\!\! \hat{x}^{svr}_{k|k-1} \!+\! K^{svr}_{k}(z_{k} \!-\! C\hat{x}^{svr}_{k|k-1})\\
%\!\!\!&=&\!\!\!\hat{x}^{svr}_{k|k-1} \!+\! K^{svr}_{k}(y_{k} \!-\! C\hat{x}^{svr}_{k|k-1})=\hat{x}_{k|k}.~~
%\end{eqnarray*}

%Notice that $L_k=-(B'S_{k+1}B+U)^{-1}B'S_{k+1}A$ and $B'S_{k+1}B+U$ has full rank.
%Either condition 1) or 2) leads to that $L_k\delta_k=0$.

%%%%%%%%%%%%%%%%%%%%%%%%%%%%%%%%%%%%%%%%%%%%%
\section{Examples} \label{sec:examples}

In this section, we provide two examples to demonstrate the performance of the privacy schemes.
\begin{example}
We consider a higher-order cooperative system with infinite-time horizon, whose parameters are given as follows:
\begin{align*}
&A=\left[
\begin{array}{ccc}
    0.19 & 0.86 & 0.10  \\
   0.31 &  0.80 &  0.44 \\
   0.13 &  0.43 &  0.40
\end{array}
\right], \qquad
 B=\left[
\begin{array}{cc}
   2.0 & 0.9 \\
   9.1 & 2.0 \\
   1.3 & 8.1 
\end{array}
\right], \\
&C=\left[
\begin{array}{ccc}
    2.0  & 1.6 & 1.2 \\
    2.0  & 2.0 & 1.1
\end{array}
\right],\qquad
Q=\left[
\begin{array}{ccc}
  1.9 & 0.9 & 0.4 \\
   0.9& 2.8 & 2.0 \\
   0.4& 2.0 & 2.4
\end{array}
\right], \\
&R=\left[
\begin{array}{cc}
   7.0 &1.8\\
   1.8 & 0.8
\end{array}
\right], \qquad
W=\left[
\begin{array}{ccc}
  1.8 & 2.0 & 0.5 \\
  2.0 & 9.8 & 0.9 \\
  0.5 &0.9 & 5.4
\end{array}
\right], \\
&U=\left[
\begin{array}{cc}
   4.5 & 1.0 \\
   1.0 & 8.8
\end{array}
\right].
\end{align*}

Given a level $\alpha$ of the LQG control performance $\mathcal{Q}_{LQG}$, the optimal privacy performance $\mathcal{Q}_{privacy}$ is determined by solving Problem \ref{problem:LQG2-infinite}. 
We denote the optimal $\mathcal{Q}_{privacy}$ as $\mathcal{Q}^\star_{privacy}$.
We plot the traces of $\mathcal{Q}^\star_{privacy}$ when $\alpha$ varies from 0 to 50 with step length of 1, which is presented in Fig. \ref{fig:LQG2-tradoff}.

\begin{figure}[htbp]
  \centering
  \vspace{-4mm}
  \includegraphics[scale=0.55]{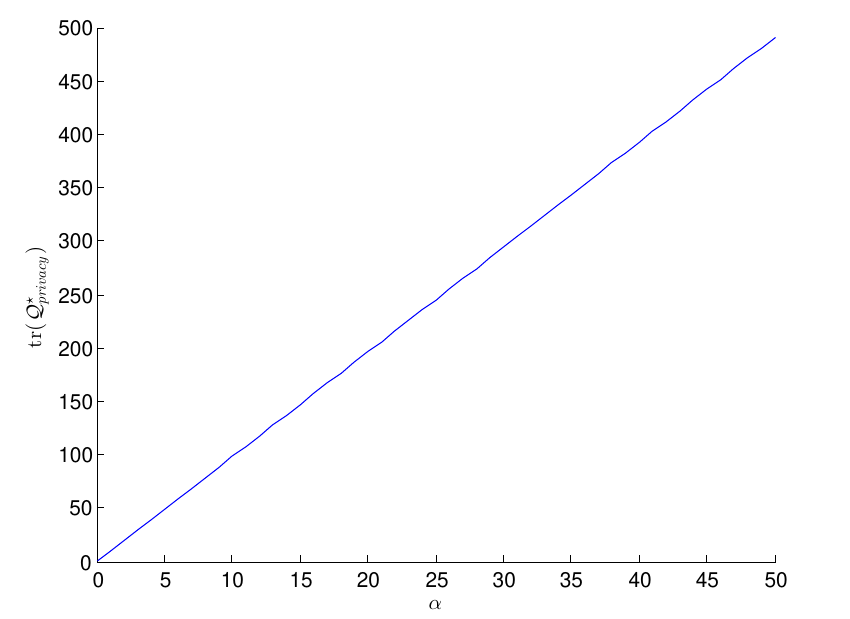}\vspace{-4mm}
  \caption{The plot of $\mathrm{tr}(\mathcal{Q}^\star_{privacy})$ when $\alpha$ varies from 0 to 50.} 
  \label{fig:LQG2-tradoff}
  \vspace{-4mm}
\end{figure}

We find the curve is a straight line. Although we cannot explicitly obtain the property of the curve in a higher-order system, it is proved in the case of scalar system, which is shown in Lemma \ref{lemma:LQG2-scalar}.

Furthermore, given a particular value of $\Sigma_\delta$, we plot the true estimates $\hat{x}_{k|k}$ and the public estimates $\hat{x}^{svr}_{k|k}$. 
Let $\Sigma_\delta=180I$. 
The plots are presented in Fig. \ref{fig:LQG2-estimates}. In the plots, we find that $\hat{x}^{svr}_{k|k}$ is deviated from $\hat{x}_{k|k}$. 

\begin{figure}[htbp]
  \centering
  \vspace{-4mm}
  \includegraphics[scale=0.55]{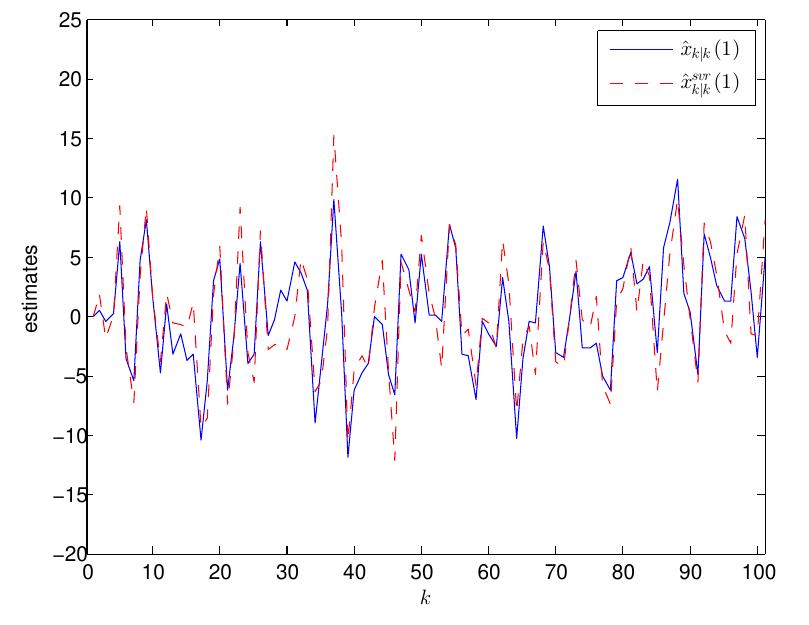}\vspace{-4mm}
  \caption{The plot of the first component of $\hat{x}_{k|k}$ and $\hat{x}^{svr}_{k|k}$. The presented time horizon is 100.} 
  \label{fig:LQG2-estimates}
  \vspace{-2mm}
\end{figure}

Moreover, we plot the traces of several covariances in Fig. \ref{fig:LQG2-covariances}: the deviation covariance $\mathcal{Q}_{privacy}^k$, its lower bound $M_k$, the true estimation error covariance $P_{k|k}$, and the real error covariance $P^{svr.real}_{k|k}$ associated with $\hat{x}^{svr}_{k|k}$. Notice that the curve of $P_{k|k}$ does not coincide with the one of $M_k$. Meanwhile, we can find that $M_k$ is only a loose bound of $\mathcal{Q}_{privacy}^k$.

\begin{figure}[htbp]
  \centering
%  \vspace{-4mm}
  \includegraphics[scale=0.55]{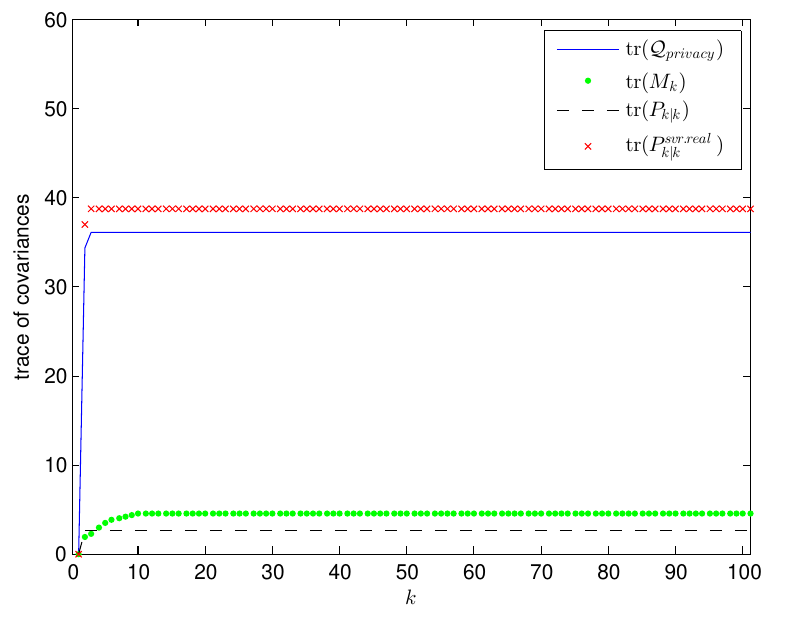}\vspace{-4mm}
  \caption{The plot of the traces of $\mathcal{Q}_{privacy}^k$, $M_k$, $P_{k|k}$, and  $P^{svr.real}_{k|k}$. The presented time horizon is 100.} 
  \label{fig:LQG2-covariances}
  \vspace{-4mm}
\end{figure}

\end{example}

%~~~~~~~~~~~~~~~~~~~~~~~~~~~~~~~~~~~~~~~~~~
\begin{example}
Consider the heating, ventilation, and air conditioning (HVAC) system of a building, which is rented and occupied by a number of user companies and organizations. 
The building's department of daily maintenance has an intelligent operation center to work for the HVAC control of the whole building. Each company could submit its requirements on environment conditions to the operation center, and then the center helps adjust the conditions for it.
However, as we know, the environment data, such as the temperature, is related to the number of occupancy individuals, and this data can be used to further deduce the position trajectory of individuals in the building, which causes leakage of privacy \cite{Jia2017ICCPS}.

Suppose one of the companies, named Company C.Y.,  is also under the service. 
This company wants to preserve its privacy when asking the operation center to do environment condition control, by a local privacy scheme proposed in this paper. 
For simplicity of presentation, we focus on the control of temperature.
In the office of Company C.Y., the temperature evolves according to eqn. (\ref{eqn:state-process}),
where each sampling time has a distance of 30 minutes, and the state $x_k$ is the difference between the real temperature $T_k$ at time $k$ and the set value of demand: $x_k=T_k-23$ (unit: degree Celsius).
The temperature sensor works according to eqn. (\ref{eqn:measurement-process}).
According to this paper, the company needs to revise the measurement data of the temperature sensor to the center. 
If the company is unable to get access to the digital signal, an easy method is to put a heating/cooling equipment around the sensor to locally change the environment.
Assume the revision function follows eqn. (\ref{eqn:LQG2-scheme}).

The company requires an LQG control of the temperature, and the parameters of the system are given as follows:
$A=1.1$, $B=3$, $C=1$, $Q=1$, $R=1$, $W=5$, $U=3$.

Consider that the company chooses a particular value of $\Sigma_\delta$. Let $\Sigma_\delta=1$. Then we plot for it the true estimates $\hat{x}_{k|k}$ and the public estimates $\hat{x}^{svr}_{k|k}$ within 24 hours.
The plots are presented in Fig. \ref{fig:LQG2-HVAC-curves}. In the plots, we find that $\hat{x}^{svr}_{k|k}$ is deviated from $\hat{x}_{k|k}$. 
The estimated temperature floats by around 2 degrees. As we need to balance the privacy preservation and the environment comfort, one can further tune the parameters to optimize the result.

%Let $\Sigma_\delta=1$. The plots are presented in Fig. \ref{fig:LQG1-states-curves}. In the plots, we find $x^{svr}_k$ is deviated from $x_k$. Meanwhile, $x_k$ is also deviated from $x^\star_k$ as the sacrifice of the preserved privacy.

\begin{figure}[htbp]
  \centering
  \vspace{-3mm}
  \includegraphics[scale=0.55]{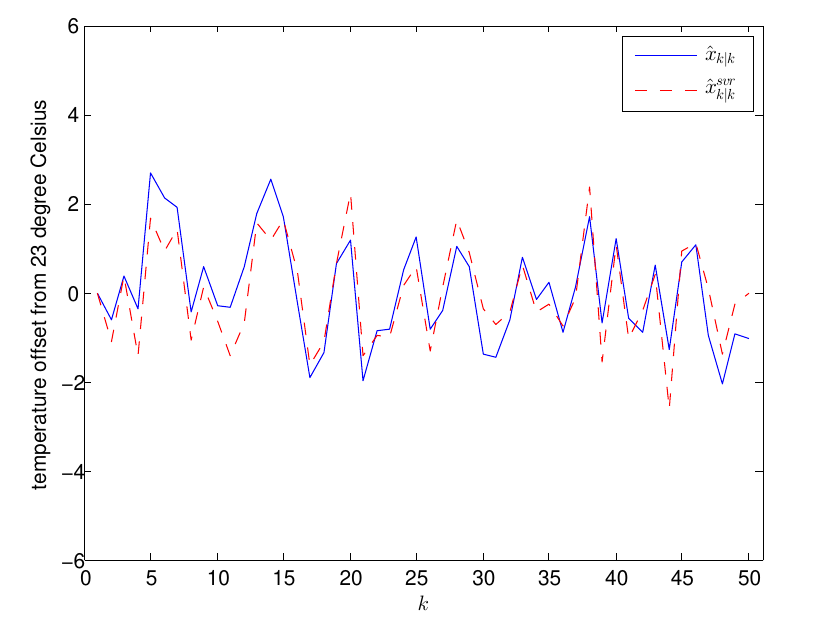}\vspace{-4mm}
  \caption{The plot of $\hat{x}_{k|k}$ and $\hat{x}^{svr}_{k|k}$. The presented time horizon is 49.} 
  \label{fig:LQG2-HVAC-curves}
  \vspace{-2mm}
\end{figure}

\end{example}

%%%%%%%%%%%%%%%%%%%%%%%%%%%%%%%%%%%%%%%%%%%%%
\section{Conclusion and Future Work}

In this paper, we propose the methodology of local design for privacy preservation for the user in a cooperative LQG control system.
%Two scenarios of the LQG control system are considered respectively: with perfect and imperfect accessible state information.
We propose the localized privacy scheme and the corresponding privacy metric.
Then we analyze the performance of the privacy preservation under the privacy scheme.
Furthermore, we analyze the service loss in LQG control performance, and show that the LQG control is asymptotically stable under the proposed privacy scheme.
At last, we propose and solve optimization problems based on the trade-off between the privacy quality and service loss.

Possible future work may be considered from the following directions.

\textbf{System structures.} 
In this paper we considered a simple cooperative networked control system. In most current literature, the considered systems also have simple structures, such as the centralized ones.
General systems in practical world should have various and more complex structure, such as multiple participating parties and more complex topologies.

\textbf{The function to process the original data.} 
In this paper, we only considered adding a zero-mean white Gaussian noise to the original measurement. More advanced forms of noise and more general functions could be further investigated.

\textbf{An extra processor at user's receiving side.} 
In this paper we only have one processor (before sending data to the server), aiming at using the simplest equipment to preserve privacy. 
%Hence, the user in unable to process the control input back from the server but to direct use it.
An extra processor at the receiving side can help process the control input back from the server, and can be expected to further improve the scheme performance. 
The cooperation of the two processors may work like an encoder and a decoder.

\textbf{Co-design with the data transmission.} 
We may call it strategic communication. 
%To use this scheme, we need to equip both encoding and decoding processors at the sending and receiving sides. 
%We use an example to describe the methodology. 
For example, the user claims a sampling frequency which doubles the real one. 
At the odd time instance, the user sends the real value and take the resulting control input; at the even time instance, the user sends interrupting signals and throws the resulting control law (by the the decoding processor). 
By doing this, the public information is quite different from the true one. %while the LQG performance is still optimal.

%Of course, the scheme of strategic communication could be extended to other privacy designs.

\textbf{The information to protect.}
Most of the literature, including our study, only considered certain signals as privacy. A larger range of concerns also could be privacy, such as the system parameters, the system structure, or even the goal of the user. The ideal relationship between the user and server may be just the server give the user what it claims wanted.

%\begin{remark}	
%\color{blue}
%Consider the general tracking problem which has a reference $\{x^r_k\}$. The reference is supposed to provide to the server or could avoid? By the privacy scheme, the realized trajectory does not entirely track it.
%
%\end{remark}

\textbf{A doubtful server.}
The possibility that the user has a local privacy scheme may motivate the curious server to employ a detection module to check the user's shared data. Then the user should consider the design of the privacy scheme under that situation.

\textbf{The trouble of the server.}
As a result of the user's local privacy scheme, the server cannot collect correct information, which will hurt its statistical consensus. Better designs benefiting both parties could be considered.

%%%%%%%%%%%%%%%%%%%%%%%%%%%%%%%%%%%%%%%%%%%%%
\appendix

\subsection{Proof of Theorem \ref{thm:LQG2-PrivacyPerformance}}
We consider the recursion of the estimates.
When $k=0$, we have
\begin{eqnarray*}
\hat{x}^{svr}_{0|0}=\hat{x}_{0|0}=\bar{x}_0,\\
P^{svr}_{0|0}=P_{0|0}=\Sigma_0.
\end{eqnarray*}
When $k=1$, the system dynamic is as follows:
\begin{eqnarray*}
x_1 \!\!\!&=&\!\!\! Ax_0+Bu_0+w_0, \\
u_0 \!\!\!&=&\!\!\! L_0\hat{x}^{svr}_{0|0},\\
y_1 \!\!\!&=&\!\!\! Cx_1+v_1.
\end{eqnarray*}
The true estimate evolves as follows:
\begin{eqnarray*}
\hat{x}_{1|0} \!\!\!&=&\!\!\!  A\hat{x}_{0|0}+Bu_{0},  \\
P_{1|0} \!\!\!&=&\!\!\!  AP_{0|0}A' + Q, \\
K_{1} \!\!\!&=&\!\!\!  P_{1|0}C'(CP_{1|0}C' + R)^{-1}, \\
\hat{x}_{1|1} \!\!\!&=&\!\!\!  \hat{x}_{1|0} + K_{1}(y_{1} - C\hat{x}_{1|0}),  \\
P_{1|1} \!\!\!&=&\!\!\!  (I -K_{1}C)P_{1|0}.
\end{eqnarray*}
At the server side, the public estimate evolves as follows: 
\begin{eqnarray*}
\hat{x}^{svr}_{1|0} \!\!\!&=&\!\!\!  A\hat{x}^{svr}_{0|0}+Bu_{0},  \\
P^{svr}_{1|0} \!\!\!&=&\!\!\!  AP^{svr}_{0|0}A' + Q, \\
K^{svr}_{1} \!\!\!&=&\!\!\!  P^{svr}_{1|0}C'(CP^{svr}_{1|0}C' + R)^{-1}, \\
\hat{x}^{svr}_{1|1} \!\!\!&=&\!\!\!  \hat{x}^{svr}_{1|0} + K^{svr}_{1}(z_{1} - C\hat{x}^{svr}_{1|0}),  \\
P^{svr}_{1|1} \!\!\!&=&\!\!\!  (I -K^{svr}_{1}C)P^{svr}_{1|0}.
\end{eqnarray*}
Meanwhile, we can find that
\begin{eqnarray*}
\hat{x}^{svr}_{1|1} \!\!\!&\neq&\!\!\!  \hat{x}_{1|1} ,  \\
K^{svr}_{1} \!\!\!&=&\!\!\!  K_{1}, \\
P^{svr}_{1|1} \!\!\!&=&\!\!\!  P_{1|1}.
\end{eqnarray*}

For general $k$, the system dynamic is
\begin{eqnarray*}
x_k \!\!\!&=&\!\!\! Ax_{k-1}+Bu_{k-1}+w_{k-1}, \\
u_{k-1} \!\!\!&=&\!\!\! L_{k-1}\hat{x}^{svr}_{k-1|k-1},\\
y_k \!\!\!&=&\!\!\! Cx_k+v_k.
\end{eqnarray*}
The true estimate evolves according to eqn. (\ref{eqn:state-time-update})-(\ref{eqn:error-covariance-measurement-update})
%\begin{eqnarray*}
%\hat{x}_{k|k-1} \!\!\!&=&\!\!\! A\hat{x}_{k-1|k-1}+Bu_{k-1}, \\
%P_{k|k-1} \!\!\!&=&\!\!\! AP_{k-1|k-1}A' + Q,  \\
%K_{k} \!\!\!&=&\!\!\! P_{k|k-1}C'(CP_{k|k-1}C' + R)^{-1},  \\
%\hat{x}_{k|k} \!\!\!&=&\!\!\! \hat{x}_{k|k-1} + K_{k}(y_{k} - C\hat{x}_{k|k-1}),  \\
%P_{k|k} \!\!\!&=&\!\!\! (I -K_{k}C)P_{k|k-1}. 
%\end{eqnarray*}
and the public estimate evolves according to eqn. (\ref{eqn:KF-pub-1})-(\ref{eqn:KF-pub-5}).
%\begin{eqnarray*}
%\hat{x}^{svr}_{k|k-1} \!\!\!&=&\!\!\! A\hat{x}^{svr}_{k-1|k-1}+Bu_{k-1}, \\
%P^{svr}_{k|k-1} \!\!\!&=&\!\!\! AP^{svr}_{k-1|k-1}A' + Q,  \\
%K^{svr}_{k} \!\!\!&=&\!\!\! P^{svr}_{k|k-1}C'(CP^{svr}_{k|k-1}C' + R)^{-1},  \\
%\hat{x}^{svr}_{k|k} \!\!\!&=&\!\!\! \hat{x}^{svr}_{k|k-1} + K^{svr}_{k}(z_{k} - C\hat{x}^{svr}_{k|k-1}),  \\
%P^{svr}_{k|k} \!\!\!&=&\!\!\! (I -K^{svr}_{k}C)P^{svr}_{k|k-1}. 
%\end{eqnarray*}
From the recursion, we can find that
\begin{eqnarray*}
\hat{x}^{svr}_{k|k} \!\!\!&\neq&\!\!\! \hat{x}_{k|k} ,  \\
K^{svr}_{k} \!\!\!&=&\!\!\! K_{k}, \\
P^{svr}_{k|k} \!\!\!&=&\!\!\! P_{k|k}.
\end{eqnarray*}
Then we compare $\hat{x}^{svr}_{k|k}$ and $\hat{x}_{k|k}$. We have
\begin{eqnarray*}
\hat{x}^{svr}_{k|k} \!\!\!&=&\!\!\! \hat{x}^{svr}_{k|k-1} + K^{svr}_{k}(z_{k} - C\hat{x}^{svr}_{k|k-1})\\
\!\!\!&=&\!\!\! (I-K^{svr}_k C) \hat{x}^{svr}_{k|k-1}+K^{svr}_{k}z_{k}, \\
\hat{x}_{k|k} \!\!\!&=&\!\!\! \hat{x}_{k|k-1} + K_{k}(y_{k} - C\hat{x}_{k|k-1})\\
\!\!\!&=&\!\!\! (I-K_k C) \hat{x}_{k|k-1}+K_{k}y_{k}.
\end{eqnarray*}
Since $K^{svr}_k=K_{k}$, 
\begin{eqnarray*}
\hat{x}^{svr}_{k|k}- \hat{x}_{k|k}\!\!\!&=&\!\!\! (I-K_k C)(\hat{x}^{svr}_{k|k-1}-\hat{x}_{k|k-1}) +K_{k}(z_{k}-y_{k}) \\
\!\!\!&=&\!\!\! (I-K_k C)(\hat{x}^{svr}_{k|k-1}-\hat{x}_{k|k-1})+K_{k}\delta_{k}.
\end{eqnarray*}
Since
\begin{eqnarray*}
\hat{x}^{svr}_{k|k-1} \!\!\!&=&\!\!\! A\hat{x}^{svr}_{k-1|k-1}+Bu_{k-1}, \\
\hat{x}_{k|k-1} \!\!\!&=&\!\!\! A\hat{x}_{k-1|k-1}+Bu_{k-1},
\end{eqnarray*}
then
\begin{eqnarray*}
\hat{x}^{svr}_{k|k-1}-\hat{x}_{k|k-1}  =  A(\hat{x}^{svr}_{k-1|k-1}-\hat{x}_{k-1|k-1}). 
\end{eqnarray*}
Hence we have
\begin{eqnarray*}
\hat{x}^{svr}_{k|k}- \hat{x}_{k|k} \!\!\!\!&=&\!\!\!\!  (I-K_k C)A(\hat{x}^{svr}_{k-1|k-1}-\hat{x}_{k-1|k-1})+K_{k}\delta_{k}
\end{eqnarray*}
and
\begin{eqnarray*}
\mathcal{Q}_{privacy}^k=(I-K_k C)A\mathcal{Q}_{privacy}^{k-1}A'(I-K_k C)' +K_k\Sigma_\delta K_k'.
\end{eqnarray*}
Since $\hat{x}^{svr}_{0|0}=\hat{x}_{0|0}$, $\mathcal{Q}_{privacy}^k$ is initialized by $\mathcal{Q}_{privacy}^0=0$.
$\hfill \blacksquare$

%-----------------------------------------------------------------------------------------------------------------------
\subsection{Proof of Theorem \ref{thm:LQG2-obj}}

In this scenario, the states are unaccessible. When we calculate the LQG objective $\mathcal{O}_{0:T}$ from backward, the expectation should condition on the information set defined as follows:
\begin{eqnarray*}
\mathbb{I}_0&=&\{\phi\},\\
\mathbb{I}_k&=&\{y_1,y_2,...,y_k,u_0,u_1,...,u_{k-1}\}, ~~~k=1,...,T.
\end{eqnarray*}
%Also define
%\begin{eqnarray*}
%H_k\!\!\!&=&\!\!\!x_k'Wx_k+u_k'Uu_k,  ~~~k=0,1,...,T-1,\\
%H_T\!\!\!&=&\!\!\!x_T' W x_T,
%\end{eqnarray*}
We also define 
\begin{eqnarray*}
H_k&=&x_k'Wx_k+u_k'Uu_k,  ~~~k=0,1,...,T-1, \label{eqn:H_k}\\
H_T&=&x_T' W x_T. \label{eqn:H_T}
\end{eqnarray*}
Then we define
\begin{eqnarray*}
\mathcal{O}_{k:T}=\bm{E}\Big(\sum_{i=k}^{T-1}H_i~\Big| ~\mathbb{I}_k\Big),
\end{eqnarray*}
where the expectation is taken with respect to all possible uncertainties, i.e., 
\begin{itemize}
\item
$x_k$,

\item
$w_k, w_{k+1},...,w_{T-1}$ and $v_{k+1}, v_{k+2},...,v_{T}$, 

\item
$\delta_k, \delta_{k+1},...,\delta_{T-1}$.
\end{itemize}
Notice that $\mathcal{O}_{k:T}$ is a function of $\mathbb{I}_k$, i.e., $\mathcal{O}_{k:T}(\mathbb{I}_k)$.
Meanwhile, $\mathcal{O}_{0:T}=\bm{E}\Big(\sum\limits_{i=0}^{T-1}H_i\Big|\mathbb{I}_0\Big)=\bm{E}\Big(\sum\limits_{i=0}^{T-1}H_i\Big)$, which is also consistent with the definition (\ref{eqn:obj2}).

%
%\textbf{About the uncertainties.} 1) The inner expectation. First, since $x_k$ is unknown now, the inner expectation is taken with respect to $x_k$. Then, $\mathcal{O}_{k+1:T}$ depends on $\mathbb{I}_{k+1}=(\mathbb{I}_{k}, u_k, y_{k+1})$. Since $y_{k+1}$ consists of $x_{k+1}$ and $v_{k+1}$, when given $\mathbb{I}_k$ and $\delta_k$, the inner expectation is also with respect to $w_k$ and $v_{k+1}$. When given $\mathbb{I}_k$ and $\delta_k$, $u_k$ is a fixed one. 2) The outer expectation. When only $\mathbb{I}_k$ is conditioned on, the randomness of $\delta_k$ in $u_k$ will be eliminated. Hence, the outer expectation is with respect to $\delta_k$.\\
%

We have 
\begin{eqnarray*}
\mathcal{O}_{k:T}\!\!\!\!&=&\!\!\!\!\bm{E}(x_k'Wx_k+u_k'Uu_k+\mathcal{O}_{k+1:T}~|~\mathbb{I}_k)\\
\!\!\!\!&=&\!\!\!\! \mathop{\bm{E}}_{\delta_k}\!\Big[\!\mathop{\expt}_{x_k,w_k,v_{k+1}}\!\!(x_k'Wx_k\!+\!u_k'Uu_k\!+\!\mathcal{O}_{k+1:T}|\mathbb{I}_k,\delta_k) \Big|\mathbb{I}_k\Big],
\end{eqnarray*}
based on which we calculate the objective $\mathcal{O}_{0:T}$ in a backward manner.
First we have
\begin{eqnarray*}
\mathcal{O}_{T:T}=\bm{E}(x_T' W x_T |~ \mathbb{I}_T).
\end{eqnarray*}
Then
\begin{eqnarray*}
\mathcal{O}_{T-1:T}\!\!\!&=&\!\!\!\bm{E}(H_{T-1}+\mathcal{O}_{T:T} ~|~\mathbb{I}_{T-1})\\ %1
\!\!\!&=&\!\!\! \mathop{\bm{E}}_{\delta_{T-1}}\Big[\mathop{\bm{E}}_{x_{T-1},w_{T-1},v_T}(x_{T-1}'Wx_{T-1}+u_{T-1}'Uu_{T-1}\\
&&\!\!\!+x_T' W x_T ~|~\mathbb{I}_{T-1},\delta_{T-1})~\big|~\mathbb{I}_{T-1}\Big].  %2
\end{eqnarray*}
To simplify the formulation, define
\begin{eqnarray*}
\mathcal{J}_{T-1:T}
\!\!\!&=&\!\!\! \mathop{\bm{E}}_{x_{T-1},w_{T-1},v_T}(x_{T-1}'Wx_{T-1}+u_{T-1}'Uu_{T-1}\\
&&\!\!\!+x_T' W x_T ~|~\mathbb{I}_{T-1},\delta_{T-1}). 
\end{eqnarray*}
Then
\begin{eqnarray*}
\mathcal{O}_{T-1:T}=
 \mathop{\bm{E}}_{\delta_{T-1}}(\mathcal{J}_{T-1:T}~|~\mathbb{I}_{T-1}).
\end{eqnarray*}
Since $\mathbb{I}_{T-1}$ and $\delta_{T-1}$ are conditioned on, the control input $u_{T-1}$, which is determined by $\mathbb{I}_{T-1}$ and $\delta_{T-1}$, is hence fixed. 
By straightforward calculation, we have
\begin{eqnarray*}
&&\!\!\!\mathcal{J}_{T-1:T}\\
\!\!\!&=&\!\!\! \bm{E}(x_{T-1}'Wx_{T-1}~|~ \mathbb{I}_{T-1},\delta_{T-1})+
u_{T-1}'Uu_{T-1}\\
 &&\!\!\!+\bm{E}\big[(Ax_{T-1}+Bu_{T-1}+w_{T-1})'W\\
&&~~~~\cdot(Ax_{T-1}+Bu_{T-1}+w_{T-1}) ~|~ \mathbb{I}_{T-1},\delta_{T-1}\big]\\  %1
%\!\!\!&=&\!\!\!  \bm{E}(x_{T-1}'Wx_{T-1}|~ \mathbb{I}_{T-1},\delta_{T-1})+
%u_{T-1}'Uu_{T-1}\\
%&&\!\!\!  +\bm{E}(x_{T-1}'A'WAx_{T-1}~|~ \mathbb{I}_{T-1},\delta_{T-1}) \\
%&&\!\!\!+u_{T-1}'B'WBu_{T-1}\!+\!\bm{E}(2x_{T-1}'A'WBu_{T-1}|~ \mathbb{I}_{T-1},\delta_{T-1})\\
%&&\!\!\! +2\bm{E}\big[(Ax_{T-1}+Bu_{T-1})'Ww_{T-1}~|~ \mathbb{I}_{T-1},\delta_{T-1}\big]\\ 
%&&\!\!\!+\bm{E}(w_{T-1}'Ww_{T-1}) \\   %2
%%\!\!\!&=&\!\!\!\bm{E}\big[x_{T-1}'(W+A'WA)x_{T-1} ~|~ \mathbb{I}_{T-1},\delta_{T-1}\big] +u_{T-1}'(U+B'WB)u_{T-1} \\
%&&+2\bm{E}(x_{T-1}' ~|~ \mathbb{I}_{T-1},\delta_{T-1})A'WBu_{T-1}+\mathrm{tr}(WQ)\\  %%3
\!\!\!&=&\!\!\! \bm{E}\big[x_{T-1}'(W+A'WA)x_{T-1} ~|~ \mathbb{I}_{T-1},\delta_{T-1}\big]\\
&&\!\!\!+u_{T-1}'(U+B'WB)u_{T-1}
 +2\hat{x}_{T-1|T-1}'A'WBu_{T-1}\\
&&\!\!\!+\mathrm{tr}(WQ)\\  %4
\!\!\!&=&\!\!\!\bm{E}\big[x_{T-1}'(W+A'WA)x_{T-1} ~|~ \mathbb{I}_{T-1},\delta_{T-1}\big]\\
&&\!\!\!+\big[u_{T-1}\!+\!(U\!+\!B'WB)^{-1}B'WA\hat{x}_{T-1|T-1}\big]'(U\!+\!B'WB)\\ 
&&\!\!\!~~\cdot\big[u_{T-1}\!+\!(U\!+\!B'WB)^{-1}B'WA\hat{x}_{T-1|T-1}\big]\\
&&\!\!\!-\hat{x}_{T\!-\!1|T\!-\!1}'A'WB(U\!+\!B'WB)^{-1}B'WA\hat{x}_{T\!-\!1|T\!-\!1}\!+\!\mathrm{tr}(WQ)\\ %5
\!\!\!&=&\!\!\!\bm{E}\Big\{x_{T-1}'\big[W\!+\!A'WA\!-\!A'WB(U\!+\!B'WB)^{-1}B'WA\big]x_{T-1}\\
&&\!\!\! ~~~|~ \mathbb{I}_{T\!-\!1},\delta_{T\!-\!1}\Big\}
+\mathrm{tr}\big[A'WB(U\!+\!B'WB)^{-1}\!B'WAP_{T\!-\!1|T\!-\!1}\big]\\
&&\!\!\!+\mathrm{tr}(WQ)\\
&&\!\!\!+\big[u_{T-1}\!+\!(U\!+\!B'WB)^{-1}B'WA\hat{x}_{T-1|T-1}\big]'(U+B'WB)\\ 
&&\!\!\!~~\cdot\big[u_{T-1}\!+\!(U\!+\!B'WB)^{-1}B'WA\hat{x}_{T-1|T-1}\big]. %6
\end{eqnarray*}
According to eqn. (\ref{eqn:LQG-S-terminal}-\ref{eqn:LQG-gain}) and  
\begin{eqnarray*}
u_{T-1}=L_{T-1}(\hat{x}_{T-1|T-1}+\hat{d}_{T-1}),
\end{eqnarray*}
we have
%\begin{eqnarray*}
%\mathcal{J}_{T-1:T}&=&\bm{E}(x_{T-1}'S_{T-1}x_{T-1}~|~ \mathbb{I}_{T-1},\delta_{T-1}) +\mathrm{tr}(S_TQ)
%+\mathrm{tr}\big[A'WB(U+B'WB)^{-1}B'WAP_{T-1|T-1}\big]\\
%&& +(u_{T-1}-L_{T-1}\hat{x}_{T-1|T-1})'(U+B'S_TB)(u_{T-1}-L_{T-1}\hat{x}_{T-1|T-1}).
%\end{eqnarray*}
\begin{eqnarray*}
\mathcal{J}_{T-1:T}
\!\!\!&=&\!\!\!\bm{E}(x_{T-1}'S_{T-1}x_{T-1}~|~ \mathbb{I}_{T-1},\delta_{T-1}) +\mathrm{tr}(S_TQ)\\
&&\!\!\!+\mathrm{tr}\big[A'WB(U+B'WB)^{-1}B'WAP_{T-1|T-1}\big]\\
&&\!\!\! +\hat{d}_{T-1}'L_{T-1}'(U+B'S_TB)L_{T-1}\hat{d}_{T-1}.
\end{eqnarray*}
Back to $\mathcal{O}_{T-1:T}$, we have
\begin{eqnarray*}
 \mathcal{O}_{T-1:T}
\!\!\!&=&\!\!\! \mathop{\bm{E}}_{\delta_{T-1}}\big(\mathcal{J}_{T-1:T}~|~\mathbb{I}_{T-1}\big)\\ %1
\!\!\!&=&\!\!\!\bm{E}(x_{T-1}'S_{T-1}x_{T-1}~|~ \mathbb{I}_{T-1}) +\mathrm{tr}(S_TQ)\\
&&\!\!\!+\mathrm{tr}\big[A'WB(U+B'WB)^{-1}B'WAP_{T-1|T-1}\big]\\
&&\!\!\!  + \bm{E}\big[\hat{d}_{T-1}'L_{T-1}'(U+B'S_TB)L_{T-1}\hat{d}_{T-1}~|~ \mathbb{I}_{T-1}\big] \\ %2
\!\!\!&=&\!\!\!\bm{E}(x_{T-1}'S_{T-1}x_{T-1}~|~ \mathbb{I}_{T-1}) +\mathrm{tr}(S_TQ)\\
&&\!\!\!+\mathrm{tr}\big[A'WB(U+B'WB)^{-1}B'WAP_{T-1|T-1}\big]\\
&&\!\!\!  + \mathrm{tr}\Big[L_{T-1}'(U+B'S_TB)L_{T-1}\bm{E}(\hat{d}_{T-1}\hat{d}_{T-1}'~|~ \mathbb{I}_{T-1})\Big] \\
\!\!\!&=&\!\!\! \bm{E}(x_{T-1}'S_{T-1}x_{T-1}~|~ \mathbb{I}_{T-1}) +\mathrm{tr}(S_TQ)\\
&&\!\!\!+\mathrm{tr}\big[A'WB(U+B'WB)^{-1}B'WAP_{T-1|T-1}\big] \\
&&\!\!\! +\mathrm{tr}\Big[L_{T-1}'(U+B'S_TB)L_{T-1}\mathcal{Q}_{privacy}^{T-1}\Big].
\end{eqnarray*}
Let 
\begin{eqnarray*}
r_T&=&0,\\
r_{T-1}&=&r_T+\mathrm{tr}(S_TQ),\\
t_T&=&0,\\
t_{T-1}&=&t_T+\mathrm{tr}\big[A'WB(U+B'WB)^{-1}B'WAP_{T-1|T-1}\big],\\
\phi_T&=&0,\\
\phi_{T-1}&=&\phi_T+\mathrm{tr}\Big[L_{T-1}'(U+B'S_TB)L_{T-1} \mathcal{Q}_{privacy}^{T-1}\Big].
%\phi_{T-1}&=&\phi_T+\mathrm{Tr} \big[A'S_TB(U+B'S_TB)^{-1}B'S_TA\Sigma_\delta\Big].
\end{eqnarray*}
Then we have
\begin{eqnarray*}
\mathcal{O}_{T-1:T}= \bm{E}(x_{T-1}'S_{T-1}x_{T-1}|~ \mathbb{I}_{T-1}) +r_{T-1}+t_{T-1}+\phi_{T-1}.
\end{eqnarray*}
Assume that $\mathcal{O}_{k+1:T}= \bm{E}(x_{k+1}'S_{k+1}x_{k+1}~|~ \mathbb{I}_{k+1})+r_{k+1}+t_{k+1}+\phi_{k+1}$.
Then
\begin{eqnarray*}
\mathcal{O}_{k:T}\!\!\!&=&\!\!\!\bm{E}(H_{k}+\mathcal{O}_{k+1:T} ~|~\mathbb{I}_{k})\\ %1
\!\!\!&=&\!\!\! \mathop{\bm{E}}_{\delta_k}\Big[\mathop{\bm{E}}_{x_k, w_k, v_{k+1}}(x_{k}'Wx_{k}+u_{k}'Uu_{k}+x_{k+1}'S_{k+1}x_{k+1}\\
&&\!\!\!~~~~~~~~~~~~~~~~~~+r_{k+1}+t_{k+1} +\phi_{k+1} ~|~\mathbb{I}_k,\delta_k)~\big|~\mathbb{I}_k\Big].
%\Big\{~\mathrm{or}  &=& \bm{E}_{\delta_k}\Big[\expt_{w_k}(x_{k}'Wx_{k}+u_{k}'Uu_{k}+x_{k+1}'S_{k+1}x_{k+1}+r_{k+1} +\phi_{k+1} |x_{k},{\color{blue}u_k})~\big|~x_k\Big].~~\Big\}
\end{eqnarray*}
Define
\begin{eqnarray*}
\mathcal{J}_{k:T}
\!\!\!&=&\!\!\! \mathop{\bm{E}}_{x_k, w_k, v_{k+1}}(x_{k}'Wx_{k}+u_{k}'Uu_{k}+x_{k+1}'S_{k+1}x_{k+1}\\
&&~~~~~~~~~~~~+r_{k+1}+t_{k+1} +\phi_{k+1} ~|~\mathbb{I}_k,\delta_k). 
\end{eqnarray*}
Then
\begin{eqnarray*}
\mathcal{O}_{k:T}
= \mathop{\bm{E}}_{\delta_{k}}\big(\mathcal{J}_{k:T}~|~\mathbb{I}_k\big).
\end{eqnarray*}
Similarly as the preceding calculation of $\mathcal{J}_{T-1:T}$,
we have
\begin{eqnarray*}
&&\!\!\!\mathcal{J}_{k:T} \\
\!\!\!&=&\!\!\!\bm{E}\big[x_k'(W\!+\!A'S_{k+1}A\\
&&~~~- A'S_{k+1}B(U\!+\!B'S_{k+1}B)^{-1}B'S_{k+1}A)x_k~|~\mathbb{I}_k,\delta_k\big]\\
&&+\mathrm{tr}\big[A'WB(U+B'WB)^{-1}B'WAP_{k|k}\big]\\
&&+\big[u_k\!+\!(U\!+\!B'S_{k+1}B)^{-1}\!B'S_{k+1}A\hat{x}_{k|k}\big]'(U\!+\!B'S_{k+1}B)\\
&&~~~\cdot \big[u_k\!+\!(U\!+\!B'S_{k+1}B)^{-1}\!B'S_{k+1}A\hat{x}_{k|k}\big]\\
&&+\mathrm{tr}(S_{k+1}Q)+r_{k+1}+t_{k+1} +\phi_{k+1}. %6
\end{eqnarray*}
Similarly, according to the definition of $S_k$ and $L_k$, and define
\begin{eqnarray*}
r_{k}\!\!\!&=&\!\!\!r_{k+1}+\mathrm{tr}(S_{k+1}Q),\\
t_{k}\!\!\!&=&\!\!\!t_T+\mathrm{tr}\big[A'S_{k+1}B(U+B'S_{k+1}B)^{-1}B'S_{k+1}AP_{k|k}\big],
\end{eqnarray*}
we have
\begin{eqnarray*}
\mathcal{J}_{k:T} \!\!\!&=&\!\!\!\bm{E}(x_k'S_{k}x_k~|~ \mathbb{I}_k,\delta_k)+r_{k}+t_k+\phi_{k+1}\\
&&\!\!\! +(u_k-L_kx_k)'(U+B'S_{k+1}B)(u_k-L_kx_k).
\end{eqnarray*}
Since
\begin{eqnarray*}
u_{k}=L_{k}(\hat{x}_{k|k}+\hat{d}_{k}),
\end{eqnarray*}
we have 
\begin{eqnarray*}
\mathcal{J}_{k:T}\!\!\!&=&\!\!\!\bm{E}(x_k'S_{k}x_k~|~ \mathbb{I}_k,\delta_k)+r_{k}+t_k+\phi_{k+1}\\
&&\!\!\!+\hat{d}_k'L_k'(U+B'S_{k+1}B)L_k\hat{d}_k.
\end{eqnarray*}
Back to $\mathcal{O}_{k:T}$, we have
\begin{eqnarray*}
\mathcal{O}_{k:T}
\!\!\!&=&\!\!\! \bm{E}(x_k'S_{k}x_k~|~ \mathbb{I}_k) +r_{k}+t_k+\phi_{k+1}\\
&&\!\!\!+\mathrm{tr}\Big[L_k'(U+B'S_{k+1}B)L_k\mathcal{Q}_{privacy}^k\Big]. \end{eqnarray*}
Let
\begin{eqnarray*}
\phi_{k}&=&\phi_{k+1}+\mathrm{tr}\Big[L_k'(U+B'S_{k+1}B)L_k \mathcal{Q}_{privacy}^k\Big],
%\phi_{T-1}&=&\phi_T+\mathrm{Tr} \big[A'S_TB(U+B'S_TB)^{-1}B'S_TA\Sigma_\delta\Big].
\end{eqnarray*}
then
\begin{eqnarray*}
\mathcal{O}_{k:T}
&=& \bm{E}(x_k'S_{k}x_k~|~ \mathbb{I}_k) +r_{k}+t_k+\phi_{k}.
\end{eqnarray*}

Hence, ultimately we have 
\begin{eqnarray*}
\mathcal{O}_{0:T}
\!\!\!&=&\!\!\!\bm{E}(x_0'S_0x_0)+r_0+t_0+\phi_0\\ %1
\!\!\!&=&\!\!\!\bm{E}(x_0'S_0x_0)+\sum_{k=0}^{T-1}\mathrm{tr}(S_{k+1}Q) + \sum_{k=0}^{T-1}\mathrm{tr}(\mathit\Phi_{k}P_k)\\
&&\!\!\!+\sum_{k=0}^{T-1}\mathrm{tr}\Big[L_k'(U+B'S_{k+1}B)L_k \mathcal{Q}_{privacy}^k\Big]\\%3
\!\!\!&=&\!\!\!\bm{E}(x_0'S_0x_0)+\sum_{k=0}^{T-1}\mathrm{tr}(S_{k+1}Q) + \sum_{k=0}^{T-1}\mathrm{tr}(\mathit\Phi_{k}P_k)\\
&&\!\!\!+\sum_{k=0}^{T-1}\mathrm{tr}\Big(\mathit\Phi_{k} \mathcal{Q}_{privacy}^k\Big).
\end{eqnarray*}
$\hfill \blacksquare$

%%%%%%%%%%%%%%%%%%%%%%%%%%%%%%%%%%%%%%%%%%%%%
\section*{Acknowledgement}

The work by Chao Yang and Wen Yang is supported by the National Natural Science Foundation of China under Grant 62336005, and the Natural Science Foundation of Shanghai under Grant 19ZR1413500.
%The work by Wen Yang is supported by National Natural Science Foundation of China (61973123), Shanghai Natural Science Foundation (18ZR1409700), and National Key R\&D Program of China (2018YFB0104400).

The work by Y. Ni is supported by the National Natural Science Foundation of China under Grant 62303196, the Natural Science Foundation of Jiangsu Province of China under Grant BK20231036, and the Basic Research Funds of Wuxi Taihu Light Project under Grant K20221005.

%%%%%%%%%%%%%%%%%%%%%%%%%%%%%%%%%%%%%%%%%%%%%
\bibliographystyle{IEEETran}
\bibliography{BibFile-2016-12}

\begin{thebibliography}{10}
\providecommand{\url}[1]{#1}
\csname url@samestyle\endcsname
\providecommand{\newblock}{\relax}
\providecommand{\bibinfo}[2]{#2}
\providecommand{\BIBentrySTDinterwordspacing}{\spaceskip=0pt\relax}
\providecommand{\BIBentryALTinterwordstretchfactor}{4}
\providecommand{\BIBentryALTinterwordspacing}{\spaceskip=\fontdimen2\font plus
\BIBentryALTinterwordstretchfactor\fontdimen3\font minus
  \fontdimen4\font\relax}
\providecommand{\BIBforeignlanguage}[2]{{%
\expandafter\ifx\csname l@#1\endcsname\relax
\typeout{** WARNING: IEEEtran.bst: No hyphenation pattern has been}%
\typeout{** loaded for the language `#1'. Using the pattern for}%
\typeout{** the default language instead.}%
\else
\language=\csname l@#1\endcsname
\fi
#2}}
\providecommand{\BIBdecl}{\relax}
\BIBdecl

\bibitem{TheSurvey}
J.~P. Hespanha, P.~Naghshtabrizi, and Y.~Xu, ``A survey of recent results in
  networked control systems,'' \emph{Proceedings of IEEE}, vol.~95, no.~1, pp.
  138--162, 2007.

\bibitem{Freidman2005}
T.~Freidman, \emph{The World is Flat}.\hskip 1em plus 0.5em minus 0.4em\relax
  New York: Farrar, Straus and Giroux, 2005.

\bibitem{AWS}
\protect{Amazon Web Service}, \url{https://aws.amazon.com/}.

\bibitem{Azure}
\protect{Microsoft Azure}, \url{https://azure.microsoft.com/}.

\bibitem{GoogleCloud}
\protect{Google Cloud}, \url{https://cloud.google.com/}.

\bibitem{Aliyun}
\protect{Alibaba Cloud}, \url{https://www.alibabacloud.com/}.

\bibitem{Goldreich2009}
O.~Goldreich, \emph{Foundations of Cryptography: Volume 2}.\hskip 1em plus
  0.5em minus 0.4em\relax Cambridge University Press, 2009.

\bibitem{Mo2012IEEE}
Y.~Mo, T.~H.-J. Kim, K.~Brancik, D.~Dickinson, H.~Lee, A.~Perrig, and
  B.~Sinopoli, ``Cyber-physical security of a smart grid infrastructure,''
  \emph{Proceedings of IEEE}, vol. 100, no.~1, pp. 195--209, 2012.

\bibitem{Krishnamurthy2018TIFS}
P.~Krishnamurthy, F.~Khorrami, R.~Karri, D.~Paul-Pena, and H.~Salehghaffari,
  ``Process-aware covert channels using physical instrumentation in
  cyber-physical systems,'' \emph{IEEE Transactions on Information Forensics
  and Security}, vol.~13, no.~11, pp. 2761--2771, 2018.

\bibitem{Agarwal2018CDC}
G.~K. Agarwal, M.~Karmoose, S.~Diggavi, C.~Fragouli, and P.~Tabuada,
  ``Distorting an adversary's view in cyber-physical systems,'' in \emph{2018
  IEEE Conference on Decision and Control (CDC)}.\hskip 1em plus 0.5em minus
  0.4em\relax IEEE, 2018, pp. 1476--1481.

\bibitem{Dwork2006ICALP}
C.~Dwork, ``Differential privacy,'' in \emph{Proceedings of the 33rd
  International Colloquium on Automata, Languages and Programming (ICALP), Part
  II}, Venice, Italy, July 2006, pp. 1--12.

\bibitem{Dwork2008survey}
------, ``Differential privacy: A survey of results,'' in \emph{Proceedings of
  the 5th International Conference on Theory and Applications of Models of
  Computation (TAMC)}, April 2008, pp. 1--19.

\bibitem{Ye2022TIFS}
D.~Ye, S.~Shen, T.~Zhu, B.~Liu, and W.~Zhou, ``One parameter
  defense—defending against data inference attacks via differential
  privacy,'' \emph{IEEE Transactions on Information Forensics and Security},
  vol.~17, 2022.

\bibitem{McSherry2007FOCS}
F.~McSherry and K.~Talwar, ``Mechanism design via differential privacy,'' in
  \emph{the 48th Annual IEEE Symposium on Foundations of Computer Science
  (FOCS'07)}.\hskip 1em plus 0.5em minus 0.4em\relax IEEE, 2007, pp. 94--103.

\bibitem{Li2020TIFS}
D.~Li, Q.~Yang, W.~Yu, D.~An, Y.~Zhang, and W.~Zhao, ``Towards differential
  privacy-based online double auction for smart grid,'' \emph{IEEE Transactions
  on Information Forensics and Security}, vol.~15, 2020.

\bibitem{LeNy2014TAC}
J.~{Le Ny} and G.~J. Pappas, ``Differentially private filtering,'' \emph{IEEE
  Transactions on Automatic Control}, vol.~59, no.~2, pp. 341--354, 2014.

\bibitem{Huang2012WPES}
Z.~Huang, S.~Mitra, and G.~Dullerud, ``Differentially private iterative
  synchronous consensus,'' in \emph{Proceedings of the 2012 ACM Workshop on
  Privacy in the Electronic Society}, Raleigh, USA, October 2012, pp. 81--90.

\bibitem{He2018TIT}
J.~He, L.~Cai, and X.~Guan, ``Preserving data-privacy with added noises:
  Optimal estimation and privacy analysis,'' \emph{IEEE Transactions on
  Information Theory}, vol.~64, no.~8, pp. 5677--5690, 2018.

\bibitem{Koufogiannis2017CDC}
F.~Koufogiannis and G.~J. Pappas, ``Differential privacy for dynamical
  sensitive data,'' in \emph{2017 IEEE 56th Annual Conference on Decision and
  Control (CDC)}.\hskip 1em plus 0.5em minus 0.4em\relax IEEE, 2017, pp.
  1118--1125.

\bibitem{LeNy2018TAC}
J.~{Le Ny} and M.~Mohammady, ``Differentially private {MIMO} filtering for
  event streams,'' \emph{IEEE Transactions on Automatic Control}, vol.~64,
  no.~1, pp. 145--157, 2018.

\bibitem{Ye2018TIT}
M.~Ye and A.~Barg, ``Optimal schemes for discrete distribution estimation under
  locally differential privacy,'' \emph{IEEE Transactions on Information
  Theory}, vol.~64, no.~8, pp. 5662--5676, 2018.

\bibitem{Han2017TAC}
S.~Han, U.~Topcu, and G.~J. Pappas, ``Differentially private distributed
  constrained optimization,'' \emph{IEEE Transactions on Automatic Control},
  vol.~62, no.~1, pp. 50--64, 2017.

\bibitem{Hale2018TCNS}
M.~T. Hale and M.~Egerstedt, ``Cloud-enabled differentially private multi-agent
  optimization with constraints,'' \emph{IEEE Transactions on Control of
  Network Systems}, vol.~5, no.~4, pp. 1693 -- 1706, December 2018.

\bibitem{Hale2018ACC}
M.~Hale, A.~Jones, and K.~Leahy, ``Privacy in feedback: The differentially
  private {LQG},'' in \emph{Proceedings of the 2018 American Control Conference
  (ACC)}, Milwaukee, USA, June 2018, pp. 3386--3391.

\bibitem{Nekouei2019Survey}
E.~Nekouei, T.~Tanaka, M.~Skoglund, and K.~H. Johansson,
  ``Information-theoretic approaches to privacy in estimation and control,''
  \emph{Annual Reviews in Control}, vol.~47, 2019.

\bibitem{Nekouei2018ACC}
E.~Nekouei, M.~Skoglund, and K.~H. Johansson, ``Privacy of information sharing
  schemes in a cloud-based multi-sensor estimation problem,'' in
  \emph{Proceedings of the 2018 American Control Conference (ACC)}, Milwaukee,
  USA, June 2018, pp. 998--1002.

\bibitem{Li2018TIT}
S.~Li, A.~Khisti, and A.~Mahajan, ``Information-theoretic privacy for smart
  metering systems with a rechargeable battery,'' \emph{IEEE Transactions on
  Information Theory}, vol.~64, no.~5, pp. 3679--3695, 2018.

\bibitem{Nekouei2018a}
E.~Nekouei, H.~Sandberg, M.~Skoglund, and K.~H. Johansson, ``Privacy-aware
  minimum error probability estimation: An entropy constrained approach,''
  arXiv:1808.02271v1, 2018.

\bibitem{Jia2017ICCPS}
R.~Jia, R.~Dong, S.~S. Sastry, and C.~J. Sapnos, ``Privacy-enhanced
  architecture for occupancy-based {HVAC} control,'' in \emph{2017 ACM/IEEE 8th
  International Conference on Cyber-Physical Systems (ICCPS)}, 2017, pp.
  177--186.

\bibitem{Tanaka2017ACC}
T.~Tanaka, M.~Skoglund, H.~Sandberg, and K.~H. Johansson, ``Directed
  information and privacy loss in cloud-based control,'' in \emph{Proceedings
  of the 2017 American Control Conference (ACC)}, Seattle, USA, May 2017, pp.
  1666--1672.

\bibitem{Farokhi2016NECSYS}
F.~Farokhi, I.~Shames, and N.~Batterham, ``Secure and private cloud-based
  control using semi-homomorphic encryption,'' in \emph{Proceedings of the 6th
  IFAC Workshop on Distributed Estimation and Control in Networked Systems},
  Tokyo, Japan, September 2016, pp. 163--168.

\bibitem{Farokhi2020}
F.~Farokhi, \emph{Privacy in Dynamical Systems}.\hskip 1em plus 0.5em minus
  0.4em\relax Springer, 2020.

\bibitem{Tran2020CEP}
J.~Tran, F.~Farokhi, M.~Cantoni, and I.~Shames, ``Implementing homomorphic
  encryption based secure feedback control,'' \emph{Control Engineering
  Practice, https://doi.org/10.1016/j.conengprac.2020.104350}, 2020.

\bibitem{Ni2021TIFS}
Y.~Ni, J.~Wu, and L.~Li, ``Multi-party dynamic state estimation that preserves
  data and model privacy,'' \emph{IEEE Transactions on Information Forensics
  and Security}, vol.~16, pp. 2288--2299, 2021.

\bibitem{Mo2014CDC}
Y.~Mo and R.~M. Murray, ``Privacy preserving average consensus,'' in
  \emph{Proceedings of the 53rd IEEE Conference on Decision and Control}, Los
  Angeles, USA, December 2014, pp. 2154--2159.

\bibitem{Mo2017TAC}
Y.~{Mo} and R.~M. {Murray}, ``Privacy preserving average consensus,''
  \emph{IEEE Transactions on Automatic Control}, vol.~62, no.~2, pp. 753--765,
  2017.

\bibitem{Yazdani2023TAC}
K.~Yazdani, A.~Jones, K.~Leahy, and M.~Hale, ``Differentially private {LQ}
  control,'' \emph{IEEE Transactions on Automatic Control}, vol.~68, no.~2, pp.
  1061--1068, 2023.

\bibitem{Alexandru2019ICCPS}
A.~B. Alexandru and G.~J. Pappas, ``Encrypted {LQG} using labeled homomorphic
  encryption,'' in \emph{Proceedings of the 10th ACM/IEEE international
  conference on cyber-physical systems}, 2019, pp. 129--140.

\bibitem{Betsekas1995}
D.~P. Bertsekas, \emph{Dynamic Programming and Optimal Control, Volume I, Third
  Edition}.\hskip 1em plus 0.5em minus 0.4em\relax Athena Scientific, 2005.

\bibitem{AndersonMoore1979}
B.~Anderson and J.~B. Moore, \emph{Optimal filtering}.\hskip 1em plus 0.5em
  minus 0.4em\relax Prentice Hall, 1979.

\bibitem{Kailath2000}
T.~Kailath, A.~H. Sayed, and B.~Hassibi, \emph{Linear Estimation}.\hskip 1em
  plus 0.5em minus 0.4em\relax Prentice Hall, 2000.

\bibitem{Boyd2004}
S.~Boyd and L.~Vandenberghe, \emph{Convex Optimization}.\hskip 1em plus 0.5em
  minus 0.4em\relax Cambridge University Press, 2004.

\bibitem{CVX}
\protect{CVX: Matlab Software for Disciplined Convex Programming},
  \url{https://cvxr.com/cvx/}.

\end{thebibliography}

\end{document}